\documentclass{article}

\usepackage{microtype}
\usepackage{graphicx}
\usepackage{subcaption}
\usepackage{booktabs} 
\usepackage{multirow}
\usepackage{hyperref}

\usepackage[arxiv]{icml2024}

\usepackage{amsmath}
\usepackage{amssymb}
\usepackage{mathtools}
\usepackage{amsthm}
\usepackage{tikz}
\usetikzlibrary{matrix,arrows,cd,calc}
\usepackage[shortlabels]{enumitem}

\usepackage{float}
\usepackage{csvsimple}
\usepackage{titlesec}

\usepackage[normalem]{ulem}
\useunder{\uline}{\ul}{}

\makeatletter
\csvset{
  autotabularcenter/.style={
    file=#1,
    before reading=\setlength{\tabcolsep}{3pt},
    after head=\csv@pretable\begin{tabular}{|*{\csv@columncount}{c|}}\csv@tablehead,
    table head=\hline\csvlinetotablerow\\\hline\hline,
    late after line=\\,
    table foot=\\\hline,
    late after last line=\csv@tablefoot\end{tabular}\csv@posttable,
    command=\csvlinetotablerow},
}
\makeatother

\usepackage[capitalize,noabbrev]{cleveref}

\theoremstyle{plain}
\newtheorem{theorem}{Theorem}[section]
\newtheorem{proposition}[theorem]{Proposition}
\newtheorem{lemma}[theorem]{Lemma}
\newtheorem{corollary}[theorem]{Corollary}
\theoremstyle{definition}
\newtheorem{definition}[theorem]{Definition}
\newtheorem{construction}[theorem]{Construction}
\newtheorem{example}[theorem]{Example}
\newtheorem{assumption}[theorem]{Assumption}
\theoremstyle{remark}
\newtheorem{remark}[theorem]{Remark}
\newtheorem*{remark*}{Remark}

\usepackage[textsize=tiny]{todonotes}

\renewcommand{\phi}{\varphi}
\renewcommand{\epsilon}{\varepsilon}

\DeclareMathAlphabet{\mathpzc}{OT1}{pzc}{m}{it}

\usepackage[mathscr]{euscript}

\makeatletter
\newcommand\DEFINEALPHABETLOOP[3]{%
  \ifx\relax#3\expandafter\@gobble\else\expandafter\@firstofone\fi
  {\expandafter\newcommand\expandafter*\csname#3#1\endcsname{#2{#3}}%
   \DEFINEALPHABETLOOP{#1}{#2}}%
}%
\newcommand\Definealphabet[2]{%
  \DEFINEALPHABETLOOP{#1}{#2}abcdefghijklmnopqrstuvwxyzABCDEFGHIJKLMNOPQRSTUVWXYZ\relax
}%
\makeatother
\Definealphabet{bb}{\mathbb}
\Definealphabet{cal}{\mathcal}
\Definealphabet{bf}{\mathbf}
\Definealphabet{sf}{\mathsf}
\Definealphabet{rm}{\mathrm}
\Definealphabet{frak}{\mathfrak}
\Definealphabet{scr}{\mathscr}
\newcommand{\Fil}{\mathsf{Fil}}
\renewcommand{\Im}{\mathsf{Im}}
\newcommand{\ord}{\mathsf{ord}}
\newcommand{\push}{\mathsf{push}}
\newcommand{\incl}{\mathsf{incl}}
\newcommand{\grid}{\mathsf{grid}}
\newcommand{\cell}{\mathsf{cell}}

\newcommand{\dsm}{\mathcal{DSM}}
\newcommand{\dm}{\mathcal{DM}}
\newcommand{\assignment}{\mathsf{assign}}
\newcommand{\cost}{\mathsf{cost}}
\newcommand{\bars}{\mathsf{bars}}

\newcommand{\VR}{\mathrm{VR}}
\newcommand{\vertex}{\mathcal{V}}
\newcommand{\R}{\mathbb{R}}
\newcommand{\loss}{\mathcal{L}}
\newcommand{\dg}{\mathrm{D}}
\newcommand{\Hil}{\mathsf{Hil}}
\newcommand{\rk}{\mathsf{rk}}
\newcommand{\OT}{\mathsf{OT}}

\renewcommand{\ker}{\mathsf{ker}}
\newcommand{\coker}{\mathsf{coker}}
\newcommand{\vect}{\mathsf{vec}_\Fbb}
\newcommand{\indexingset}{\mathsf{indMass}}
\newcommand{\ordHil}{\mathsf{sortHil}}
\newcommand{\itsgridincl}{\mathsf{itsIncl}}

\def\papertitle{Differentiability and Optimization of Multiparameter Persistent Homology}

\icmltitlerunning{\papertitle}

\begin{document}

\twocolumn[
\icmltitle{\papertitle}



\icmlsetsymbol{equal}{*}

\begin{icmlauthorlist}
\icmlauthor{Luis Scoccola}{equal,ox}
\icmlauthor{Siddharth Setlur}{equal,eth}
\icmlauthor{David Loiseaux}{datashape}
\icmlauthor{Mathieu Carrière}{datashape}
\icmlauthor{Steve Oudot}{steve}
\end{icmlauthorlist}

\icmlaffiliation{ox}{Mathematical Institute, University of Oxford, UK}
\icmlaffiliation{eth}{Department of Mathematics, ETH Zürich, Switzerland}
\icmlaffiliation{datashape}{DataShape, Centre Inria d’Université Côte d’Azur, France}
\icmlaffiliation{steve}{GeomeriX, Inria Saclay and École polytechnique, Paris, France}

\icmlcorrespondingauthor{Luis Scoccola}{luis.scoccola@maths.ox.ac.uk}
\icmlcorrespondingauthor{Siddharth Setlur}{ssetlur@student.ethz.ch}

\icmlkeywords{Machine Learning, ICML}

\vskip 0.3in
]



\printAffiliationsAndNotice{\icmlEqualContribution} 


\begin{abstract}
Real-valued functions on geometric data---such as node attributes on a graph---can be optimized using descriptors from persistent homology, allowing the user to incorporate topological terms in the loss function.
When optimizing a single real-valued function (the one-parameter setting), there is a canonical choice of descriptor for persistent homology: the barcode.
The operation mapping a real-valued function to its barcode is differentiable almost everywhere, and the convergence of gradient descent for losses using barcodes is relatively well understood.
When optimizing a vector-valued function (the multiparameter setting), there is no unique choice of descriptor for multiparameter persistent homology, and many distinct descriptors have been proposed. 
This calls for the development of a general framework for differentiability and optimization that applies to a wide range of multiparameter homological descriptors.
In this article, we develop such a framework and show that it encompasses well-known descriptors of different flavors, such as signed barcodes and the multiparameter persistence landscape.
We complement the theory with numerical experiments supporting the idea that optimizing multiparameter homological descriptors can lead to improved performances compared to optimizing one-parameter descriptors, even when using the simplest and most efficiently computable multiparameter descriptors.
\end{abstract}

%

\newpage

\section{Introduction}
\label{sec:introduction}

\noindent\textbf{Context.}
Persistent homology (PH), the main tool of topological data analysis (TDA), produces a descriptor in the form of a {\em barcode} for real-valued {\em filtering functions}, i.e., real-valued functions on graphs (more generally, simplicial complexes) that assign higher values to edges than to their vertices (more generally, monotonic with respect to the face order). PH has found a variety of applications in machine learning (ML), notably:
as a featurization technique 
encoding complementary information to traditional descriptors (see, e.g., \cite{hensel} for a survey), offering stability guarantees with respect to perturbations of the data~\cite{Cohen-Steiner2007}, and the possibility for end-to-end feature learning~\cite{hofer-2, Carriere2021a};
as a topological regularizer for constraining models to follow some prescribed topology in order to avoid, e.g., overfitting~\cite{chen, Moor2020};
and as a topological layer in neural networks for enhancing network performance~\cite{Carriere2020, Zhao2020, Kim2020}.
A sound mathematical foundation for this type of applications, and for barcode optimization in general, was developed in~\cite{Carriere2021a}.

Recent contributions in the area of TDA for ML have demonstrated empirically the added value, in terms of learning performances, of replacing real-valued filtering functions by $\R^n$-valued filtering functions in the pipelines involving~PH; see, e.g., \cite{Carriere2020b,todd,Loiseaux2023,emp}.
However, in contrast to the usual real-valued setting, there does not exist---and, due to fundamental algebraic reasons, there \emph{cannot} exist \cite{carlsson-zomorodian,bauer-scoccola}---any discrete descriptor like the barcode that completely encodes the PH of $\R^n$-valued filtering functions. As a consequence, the aforementioned contributions resorted to a variety of incomplete descriptors, such as {\em multiparameter persistence landscapes}~\cite{Vipond2020} and generalizations \cite{Xin2023a}, {\em signed barcodes as measures}~\cite{Loiseaux2023}, {\em multiparameter persistent images}~\cite{Carriere2020b, Loiseaux2023b}, or {\em multiparameter persistence kernels}~\cite{Corbet2019}.

In order to generalize the use of these incomplete descriptors to other applications in machine learning, 
it is necessary to extend the existing optimization framework to them.
To the best of our knowledge, this has not been done yet, but we expect to see it happen in the near future. However, the diversity of the proposed descriptors is likely to induce the parallel development of multiple competing, specialized frameworks,
depriving the field of a general, unified theory. 

In order to prevent this undesirable situation we aim at developing optimization frameworks for classes of descriptors instead of single descriptors.  
The problem is then to find the right level of abstraction to prove a differentiability and convergence result, that is, a general enough -- yet still easily checkable -- set of assumptions, under which these seemingly very different invariants can be proven to be differentiable with gradient descent convergence guarantees.

%
%
The present paper is a first attempt at this, focusing on the class of descriptors characterized as being {\em semilinearly determined on grids}. A formal statement of this condition is given in \cref{section:theory} (\cref{def:semilinearly-determined}), but intuitively, it means that the descriptor, viewed as a map from the space~$\Fil_n(K)$ (of $\R^n$-valued filtering functions on a fixed simplicial complex~$K$) to a Euclidean space~$\R^D$ factors through a simple linear space consisting of inclusions of grids in $\Rbb^n$, where it behaves semilinearly (in the sense of piecewise affinely).
This property is directly inspired from the well-established fact that the PH of real-valued filtering functions itself can be computed over finite integer grids, and behaves affinely with respect to grid inclusions: indeed, the barcodes produced by the persistence algorithm~\cite{elz-tps-02,zc-cph-05} depend affinely on the input function values as long as the pairing of the simplices in~$K$ remains fixed, and the pairing itself is entirely prescribed by  the pre-order on the simplices induced by the function values.
While we cannot assert that all future descriptors for $\R^n$-valued filtering functions will satisfy our condition, we show that currently known descriptors of very different flavors, such as the signed barcodes as measures and the multiparameter persistence landscapes, do.

\smallskip \noindent\textbf{Theoretical contributions.}
\cref{theorem:main-results}, below, summarizes our contributions regarding general homological descriptors of multiparameter filtrations (\cref{theorem:grid-determined-implies-semialgebraic,proposition:gradient}).
It generalizes the existing optimization framework for barcodes of real-valued filtering functions~\cite{Carriere2021a} to descriptors of $\R^n$-valued filtering functions that are semilinearly determined on grids.
\begin{theorem}
    \label{theorem:main-results}
    Assume given a simplicial complex $K$, a number of parameters $n \in \Nbb_{\geq 1}$, and a descriptor $\alpha\colon \Fil_n(K) \to \R^D$. If $\alpha$ is semilinearly determined on grids, then it is a semilinear map,
    with explicit Clarke subdifferential.
    If furthermore $\alpha$ is included in an optimization pipeline of the following form, where $\Phi$ is a parametrized family of fitrations and where $E$ is some loss function:
    \[
        \Rbb^d \xrightarrow{\;\;\;\Phi\;\;\;} \Fil_n(K) \xrightarrow{\;\;\;\alpha\;\;\;} \Rbb^D \xrightarrow{\;\;\;E\;\;\;} \Rbb,
    \]
    such that
    the composite objective function
    $\Lcal \coloneqq E \circ \alpha \circ \Phi$ is locally Lipschitz, and
    the maps $\Phi$ and $E$ are definable over a common o-minimal structure,
    then, under the usual assumptions of the stochastic subgradient method (\cref{assumption:davis}), almost surely, every limit point of the iterates of the stochastic subgradient method is critical for the objective function~$\Lcal$, and the sequence of  values of~$\Lcal$ converges.
\end{theorem}

For illustrative purposes, \cref{theorem:main-results} is stated for a pipeline using a single (topology-based) loss.
Several generalizations are possible and we comment on these in \cref{remark:generalizations-of-main-theorem}.

To illustrate the flexibility of our approach, we apply our framework to two widely different topology-based descriptors: the signed barcodes as measures~\cite{Loiseaux2023} derived from the Hilbert function \cite{oudot-scoccola} and the multiparameter persistence landscape~\cite{Vipond2020}.
As in the one-parameter case, these descriptors take values in infinite dimensional spaces, and thus we prove semilinearity for suitable finite dimensional representations of them (see \cref{section:semilinearity-of-known-invariants} for details).

\begin{theorem}
    \label{theorem:semilinearity-main-theorem}
    The sorted Hilbert decomposition and the evaluated multiparameter persistence landscape are semilinearly determined on grids.
\end{theorem}

The locally Lipschitzness condition of \cref{theorem:main-results} can then be satisfied using \cref{theorem:satisfy-locally-lipschitz-condition}, in the case of the Hilbert decomposition signed measure, and
the stability theorem in \cite{Vipond2020}, in the case of multiparameter persistence landscapes.

Analogous results can be obtained for other invariants, such as for instance signed measures derived from the rank invariant (\cref{definition:rank-signed-measure}).

See \cref{figure:dependency-graph} for a dependency graph of the main definitions, results, and experiments in the paper, and \cref{tab:notation-table} for a notation table.


\smallskip \noindent\textbf{Practical contributions.}
\cref{example:distance-to-a-measure,example:integration-function}, in \cref{section:applications}, give concrete instantiations of our theory suited for end-to-end optimization using signed barcodes as measures.
In \cref{sec:experiments} we report on numerical experiments showcasing these and other pipelines, on both synthetic and real data.
The main conclusion is that, for machine learning purposes, multiparameter homological descriptors tend to outperform their one-parameter counterpart.

%

\smallskip \noindent\textbf{Related work.}
As already mentioned, and to the best of our knowledge, the general problem of establishing the differentiability and gradient descent convergence of multiparameter persistence has not been addressed in the literature so far.
Concurrent work \cite{dgril} addresses this problem in the special case of GRIL \cite{Xin2023a}, a generalization of the multiparameter persistence landscape, and develops a specific, yet elementary, optimization framework for GRIL.

The one-parameter case is well studied \cite{Carriere2021a, leygonie, leygonie-2},
with applications including shape matching and classification \cite{hofer-4,poulenard},
representation learning \cite{hofer},
graph classification \cite{hofer-2,horn,zhao-ye-chen-wang},
and regularization \cite{chen,hofer-3};
see also Sections~3.2.3~and~3.3.1 of \cite{hensel} for an overview, but note that, at this point, any list will necessarily be incomplete.

\section{Background}
\label{sec:background}

\begin{figure}
    \centering
    \includegraphics[width=\linewidth]{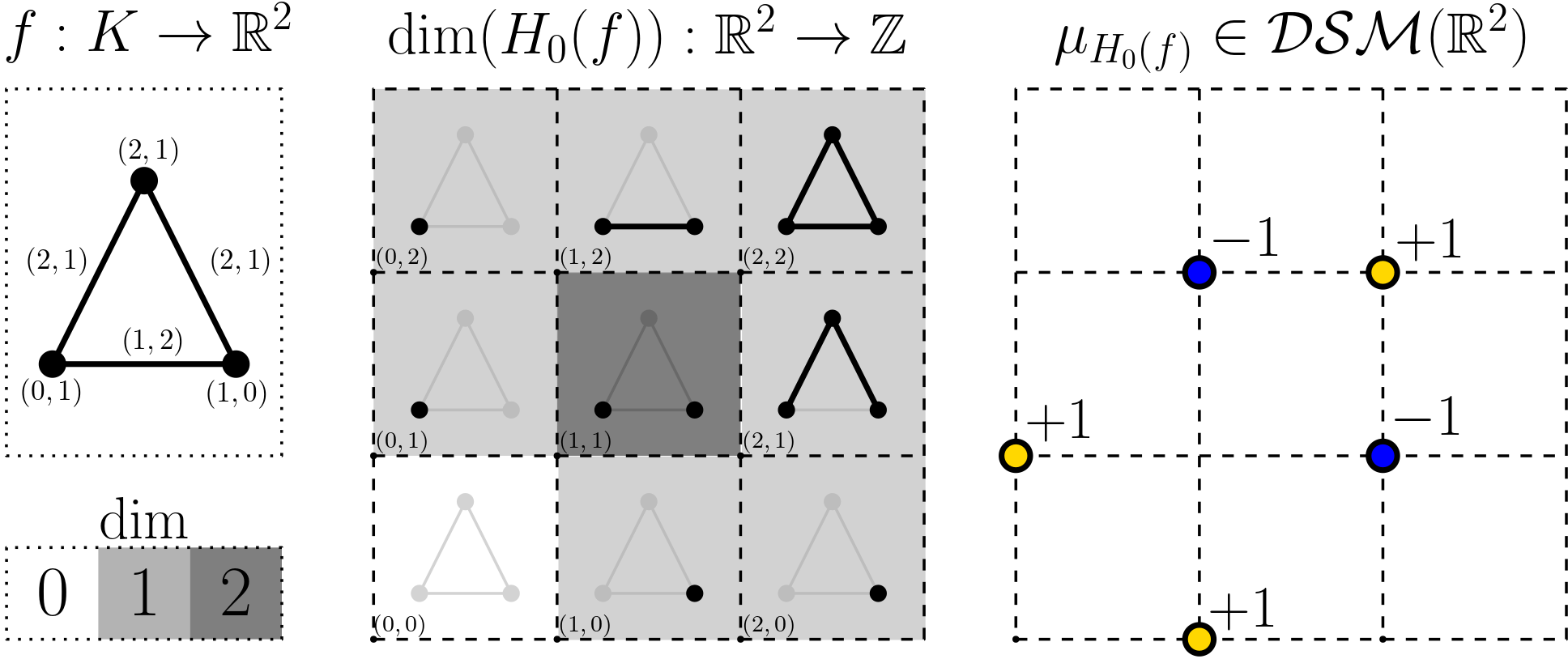}
    \caption{An example, from \cite{Loiseaux2023}, of a filtered simplicial complex, the Hilbert function of its $0$th persistent homology, and the corresponding Hilbert decomposition signed measure.}
    \label{figure:basic-concepts}
\end{figure}

This section can be skipped and only referred to as required, depending on the background of the reader.
For clarity and conciseness, we do not give full details of some constructions that are not used in the main results of this paper (see \cref{section:more-background} for more background).
Whenever we skip details we warn the reader and provide a precise reference.

\smallskip \noindent\textbf{Basic notation.}
If $n \geq 1 \in \Nbb$ and $X$ is a finite set, we let $(\Rbb^n)^X$ \label{functions-finite-set} denote the set of functions $X \to \Rbb^n$, which is a real vector space of dimension $n \times |X|$.

\smallskip \noindent\textbf{Simplicial complexes.}
A finite \emph{simplicial complex} $K$\label{simp-comp} consists of a finite set $X$ and a collection $K \subseteq \mathsf{parts}(X) \setminus \{\emptyset\}$ of non-empty subsets $X$ with the property that $\{x\} \in K$ for every $x \in X$, and if $\sigma \in K$ and $\tau \subseteq \sigma$, then $\tau \in K$.
Let $i \in \Nbb$; the $i$-simplices of $K$ are the sets $\sigma \in K$ such that $|\sigma|=i+1$.
In particular, the $0$-simplices correspond exactly to the elements of the underlying set $X$.
For this reason, it is common to denote a simplicial complex by $K$, leaving the underlying set $X$ implicit.

Note that any simplicial complex $K$ is automatically a partially ordered set where $\tau \leq \sigma \in K$ precisely when $\tau \subseteq \sigma$.
This order is known as the \emph{face order} of $K$.

\begin{example}
    Any finite, simple, undirected graph $G$ is equivalently a simplicial complex with $0$-simplices the vertices of $G$, $1$-simplices the edges of $G$, and no higher-dimensional simplices.
\end{example}

\smallskip \noindent\textbf{Filtrations.}
Let $K$ be a finite simplicial complex.

\begin{definition}\label{definition:n-filtrations}
    Let $n \in \Nbb$.
    An \emph{$n$-filtration} on $K$ consists of a function $f : K \to \Rbb^n$, which is monotonic with respect to the face order of $K$ and product (partial) order on $\Rbb^n$.
\end{definition}

Unraveling definitions, an $n$-filtration is a function $f : K \to \Rbb^n$ mapping simplices of $K$ to vectors in $\Rbb^n$, with the property that for each pair of simplices $\sigma \subseteq \tau \in K$ and $1 \leq i \leq n$, we have $f_i(\sigma) \leq f_i(\tau)$, where $f_i : K \to \Rbb$ denotes the $i$th coordinate of $f$.
See \cref{figure:basic-concepts} for an illustration.

We let $\Fil_n(K)$ denote the set of $n$-filtrations of $K$.
Note that $\Fil_n(K) \subseteq (\Rbb^n)^K$.

\begin{example}
    Let $G$ be an undirected graph and let $f' : \mathsf{vert}(G) \to \Rbb^n$ be any function.
    One can extend $f'$ to a filtering function $f \in \Fil_n(G)$ by mapping a $0$-simplex $\{x\}$ (corresponding to a vertex $x$) to $f'(x) \in \Rbb^n$, and a $1$-simplex $\{x,y\}$ (corresponding to an edge between $x$ and $y$) to $(\max(f'(x)_1, f'(y)_1), \dots, \max(f'(x)_n, f'(y)_n)) \in \Rbb^n$.
    This is an instance of a lower-star filtration \cite{edelsbrunner-harer}, a
    common construction in TDA.
\end{example}


\smallskip \noindent\textbf{(Persistent) homology.}
For detailed introductions to one- and multiparameter persistence, see, e.g., \cite{edelsbrunner-harer} and \cite{botnan-lesnick}.
Fix a field $\Fbb$ for the rest of this paper, let $i \in \Nbb$, and let $K$ be a simplicial complex.

The $i$th \emph{homology} of $K$ with coefficients in $\Fbb$ is an $\Fbb$-vector space $H_i(K)$; see, e.g., \cite{fomenko} for a precise definition with illustrations.
Informally, the dimension of $H_i(K)$ counts the number of independent $i$-dimensional holes of $K$; as an example, the dimension of $H_0(K)$ equals the number of connected components of $K$.

The homology construction is functorial, which in particular implies that if $K \subseteq K'$ are simplicial complexes, there is an $\Fbb$-linear map $H_i(K) \to H_i(K')$ of $\Fbb$-vector spaces.
Then, given $f \in \Fil_n(K)$, we can filter $K$ as follows: for $r \in \Rbb^n$, let $K^f_r \coloneqq \{\sigma \in K : f(\sigma)\leq r\} \subseteq K$.
Since $f$ is a filtration, it's clear that $K^f_r \subseteq K^f_s$ whenever $r \leq s \in \Rbb^n$, where, again, $\leq$ denotes the product (partial) order on $\Rbb^n$.
This allows for the following construction.
For convenience of notations later on, we let $H_i(f)(r) \coloneqq H_i(K^f_r)$.

\begin{definition}\label{definition:multipara-pers-hom}
    Let $f \in \Fil_n(K)$ and let $i \in \Nbb$.
    The $i$th \emph{multiparameter persistent homology} of $f$, denoted $H_i(f)$, consists of the collection of vector spaces $\{H_i(f)(r)\}_{r \in \Rbb^n}$ and the linear maps $\{\phi^f_{r,s} : H_i(f)(r) \to H_i(f)(s)\}_{r \leq s \in \Rbb^n}$, known as the \emph{structure maps}.
\end{definition}

See \cref{figure:basic-concepts} for an illustration.

Let $f \in \Fil_n(K)$.
The $i$th \emph{Hilbert function}\label{hil-func} of $f$ is the function $\Hil(H_i(f)) : \Rbb^n \to \Zbb$ defined as the pointwise dimension of $H_i(f)$, that is
\[
    \Hil(H_i(f))(r) \coloneqq \dim(H_i(f)(r)).
\]
See \cref{figure:basic-concepts} for an illustration.

The $i$th \emph{rank invariant}\label{rank-inv} of $f$ is the function $\rk(H_i(f)) : \{(r,s) \in (\Rbb^n)^2 : r \leq s\} \to \Zbb$ defined as the rank of the structure maps of $H_i(f_i)$, that is
\[
    \rk(H_i(f))(r,s) \coloneqq \rk(\phi^f_{r,s}).
\]
The rank invariant has no less information than the Hilbert function: $\rk(H_i(f))(r,r) = \Hil(H_i(f))(r)$ for $r \in \Rbb^n$.

A key motivation for using homology-based descriptors is that they are automatically isomorphism-invariant (\cref{section:iso-invariance}), meaning that they do not depend on, e.g., arbitrary labelings of data.

\smallskip \noindent\textbf{Discrete signed measures.}
We define discrete signed measures,
which are needed for the definition of the Hilbert decomposition signed measure \cite{Loiseaux2023}.

If $M$ is a metric space, the set of \emph{discrete measures}\label{disc-measures} on $M$, denoted $\dm(M)$, is the set of all finite, positive, integer linear combination of Dirac masses on $M$.
The set of discrete signed measures on $M$, denoted $\dsm(M)$, consists of the set of finite, non-necessarily positive, integer linear combinations of Dirac masses on $M$.
As usual, one can endow the set of discrete signed measures on $M$ the optimal transport distance (\cref{section:pot-definition}), denoted $\OT$, which is an extended metric.
If $\mu \in \dsm(M)$, we let $\mu^+$ and $\mu^-$ denote its positive and negative part of $\mu$, respectively, according to its Jordan decomposition (see, e.g., p.~421 of \cite{billingsley}).

\smallskip \noindent\textbf{Homological descriptors of multiparameter filtrations.}
A \emph{descriptor}\label{discriptor} of filtrations of $K$ consists of a set $A$ and a function $\Fil_n(K) \to A$.

We now define the Hilbert decomposition signed measure (originally introduced in \cite{Loiseaux2023}), which is a descriptor of multiparameter filtrations.
This descriptor completely characterizes the Hilbert function of a filtration as a discrete signed measure on $\Rbb^n$; this is remarkable, since the Hilbert function is an element of the space of functions $\Rbb^n \to \Rbb$, an infinite dimensional space.
For the definition of the rank decomposition signed measure, a stronger descriptor that strictly generalizes the one-parameter barcode, and that completely characterizes the rank invariant of filtrations, see \cref{section:rank-dec-signed-measure}.

\begin{definition}
    \label{definition:hilbert-sm-rk-sm}
    Let $f \in \Fil_n(K)$ and let $i \in \Nbb$.
    The $i$th \emph{Hilbert decomposition signed measure} is the unique measure $\mu^\Hil_{H_i(f)} \in \dsm(\Rbb^n)$ such that, for all $r \in \Rbb^n$,
    \[
        \dim(H_i(f)(r)) = \mu^\Hil_{H_i(f)}\left(\{s \in \Rbb^n : s \leq r\}\right).
    \]
\end{definition}
In order to make $\Rbb^n$ into a metric space, so that the optimal transport distance on $\dsm(\Rbb^n)$ is defined, we use $\|-\|_\infty$.

See \cref{figure:basic-concepts} for an illustration.

\smallskip \noindent\textbf{Semilinear functions and o-minimal structures.}
We now introduce semilinear and semialgebraic functions.
A generalization of these concepts is that of functions that are defined over an o-minimal structure.
For details about this theory, we refer the reader to \cite{vandenDries}.
We care about these types of well-behaved functions because gradient descent can be applied on them with convergence guarantees  \cite{davis-drusvyatskiy-kakade-lee}.

A set $S \subseteq \Rbb^n$ is \emph{semilinear}\label{semilinear-sets} (resp.~\emph{semialgebraic}) if it is a finite union of sets of the form $\{(x_1, \dots, x_n) \in \Rbb^n : P(x_1, \dots, x_n) = 0\}$ with $P$ a polynomial degree $\leq 1$ (resp.~of any degree), and sets of the form $\{(x_1, \dots, x_n) \in \Rbb^n : P(x_1, \dots, x_n) > 0\}$, with $P$ a polynomial degree $\leq 1$ (resp.~of any degree).
If $S \subseteq \Rbb^n$ is semilinear (resp.~semialgebraic), a function $\psi : S \to \Rbb^m$ is \emph{semilinear} (resp.~\emph{semialgebraic}) if its graph $\{(x,\psi(x)) : x \in S\} \subseteq \Rbb^n \times \Rbb^m$ is a semilinear (resp.~semialgebraic) set.
Note that the definitions above also make sense when $\Rbb^n$ and $\Rbb^m$ are replaced with any real vector spaces of finite dimension, since any such vector space can be identified with $\Rbb^n$ using an arbitrary linear isomorphism, and any choice of linear isomorphism leads to the same notions of semilinearity (resp.~semialgebraicity).

\begin{example}
Any piecewise linear function is semilinear.
For example, any neural network with ReLU activation functions models a semilinear function.
Note, however, that, in general, semilinear functions need not be continuous.
\end{example}

%

\smallskip \noindent\textbf{The stochastic subgradient method.}
Let $\Lcal : \Rbb^d \to \Rbb$ be differentiable almost everywhere; this is automatically the case if $\Lcal$ is locally Lipschitz, by Rademacher's theorem (see, e.g., \cite{evans-gariepy}).
The \emph{Clarke subdifferential} \label{clarke-subdiff} \cite{clarke} of $\Lcal$ at $z \in \Rbb^d$ is
\[
    \partial \Lcal(z) \coloneqq \mathsf{ConvHull} \left\{\lim_{z_k \to z} \nabla \Lcal(z_k) : \Lcal \text{ diff.~at $\{z_k\}_{k \in \Nbb}$ }\right\}
\]
and the point $z$ is \emph{critical} if $0 \in \partial\Lcal(z)$.

In the situation above, one can perform the \emph{stochastic subgradient method} by choosing a learning rate in the form of a sequence $\{a_k \in \Rbb\}_{k \in \Nbb}$,
a sequence $\{\zeta_k\}_{k \in \Nbb}$ of real random variables,
and an initialization $x_0 \in \Rbb^d$,
and defining $\{x_k\}_{k \in \Nbb}$ for $k \geq 1$ recursively by
\[
    x_{k+1} \coloneqq x_k - a_k (y_k + \zeta_k), \text{with } y_k \in \partial \Lcal(x_k).
\]

See \cref{assumption:davis} for standard assumptions about the choices involved in the stochastic subgradient method, as taken from Assumption~C of \cite{davis-drusvyatskiy-kakade-lee}.
And see \cref{remark:satisfying-davis}, which addresses satisfying these assumptions.

\section{Theoretical contributions}
\label{section:theory}

\subsection{Stratifying the space of multifiltrations}
\label{section:stratifying}

In this section, we partition the space $\Fil_n(K)$ into finitely many affine regions, which we call cells, and use this to prove \cref{corollary:semialgebraic-map-from-fil}, which provides an easily verifiable condition for a function with domain $\Fil_n(K)$ to be semilinear.
We also show that cells are linearly diffeomorphic to open, convex, semilinear sets of a Euclidean space.

\cref{figure:illustration-main-concepts} illustrates some of the main concepts introduced in this section.

\begin{figure*}
    \centering
    \includegraphics[width=0.85\linewidth]{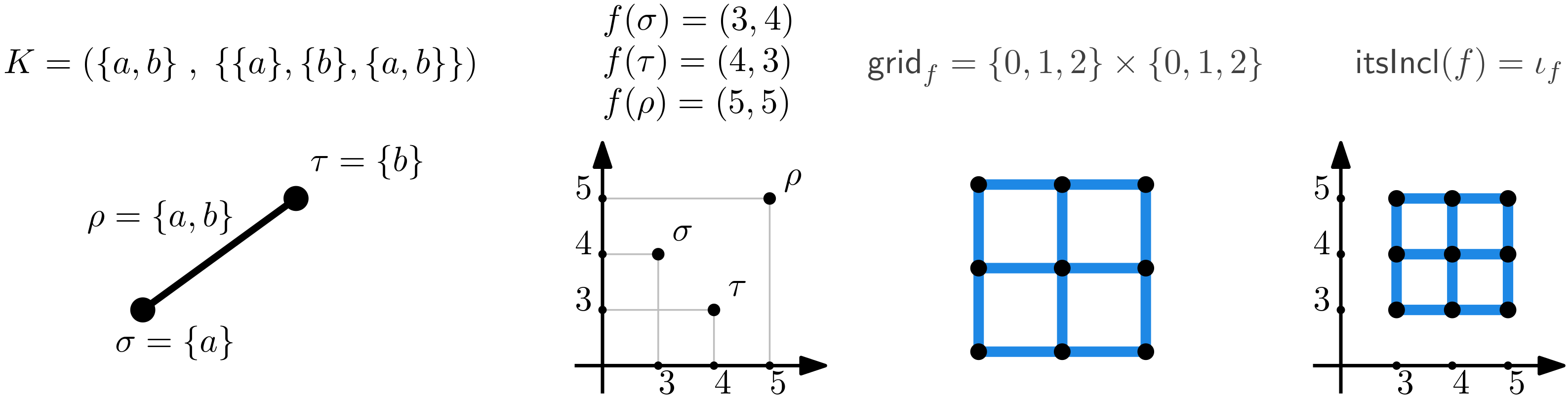}
    \caption{An example of a simplicial complex $K$, a two-parameter filtration $f$ of $K$, the grid of $f$ according to \cref{construction:grid-from-filtration}, and the aligned grid inclusion associated with $f$ according to \cref{construction:grid-from-filtration,definition:iota}.}
    \label{figure:illustration-main-concepts}
\end{figure*}

\subsubsection{Cells}
In order to define cells, we first define grids, a very simple type of finite poset, and grids obtained from filtrations.

Given $m \in \Nbb$, we let $[m] = \{0 < 1 < \cdots < m-1\}$ denote the corresponding linear order.

\begin{definition}
    \label{definition:grid}
A \emph{grid} $\Gcal$ is any product poset $\Gcal = \Gcal_1 \times \cdots \times \Gcal_n$, with $\Gcal_i = [m_i]$ for some $m_1, \dots, m_n \geq 1 \in \Nbb$.
\end{definition}

Note that, given a function $f : X \to \Rbb$, there exists a unique triple $(m \in \Nbb, \ord_f : X \to [m], \iota_f : [m] \to \Rbb)$ such that $\ord_f$ is surjective, $\iota_f$ is monotonic and injective, and $f = \iota_f \circ \ord_f$.
Indeed, $\ord_f$ represents the unique linear pre-order of $X$ induced by $\Im(f) \subseteq \Rbb$, so that $m = |\Im(f)|$.

\begin{construction}
    \label{construction:grid-from-filtration}
Given $f : K \to \Rbb^n$, let $m_i = |f_i(K)|$.
Let $\grid_f = [m_1] \times \cdots \times [m_n]$, and define
\begin{align*}
    \ord_f  \coloneqq \ord_{f_1(K)} \times \cdots \times \ord_{f_n(K)} &: K \to \grid_f\\
    \iota_f \coloneqq \iota_{f_1(K)} \times \cdots \times \iota_{f_n(K)} &:  \grid_f \to \Rbb^n
\end{align*}
\end{construction}

Note that $f = \iota_f \circ \ord_f : K \to \Rbb^n$.

We can now define the cells of $\Fil_n(K)$, which, informally, are sets of filtrations that induce the same preorder on the simplices of $K$.

\begin{definition}
    \label{definition:cell}
    The \emph{cell} of $f \in \Fil_n(K)$, denoted $\cell(f) \subseteq \Fil_n(K)$, is the set of all $g \in (\Rbb^n)^K$ such that $\grid_f = \grid_g$ and $\ord_f = \ord_g$.
\end{definition}

\begin{remark*}
    Another reasonable definition of cell would be to directly draw from the one-parameter case and declare the cell of a filtration $f : K \to \mathbb{R}^n$ to be the set of filtrations $g : K \to \mathbb{R}^n$ that induce the same preorder on the simplices of $K$ as $f$ does; the difference with our choice being that we require each coordinate $g_i : K \to \mathbb{R}$ to induce the same preorder on the simplices of $K$ as $f_i$.
    For the purposes of this remark, let us call this $\cell'(f)$.
    It can be easily seen that $\cell'(f)$ is a disjoint union of cells in the sense of \cref{definition:cell}.
    But, importantly, $\cell'(f)$ is typically not convex, or even connected: take $K$ consisting of two points and no other simplices, then $g : K \to \mathbb{R}^2$ mapping the first point to $(0,1)$ and the second to $(1,0)$ is in $\cell'(f)$, where $f$ maps the first point to $(1,0)$ and the second to $(0,1)$; but $f$ and $g$ cannot be connected through a path that stays within $\cell'(f)$.
    It can also be checked that usual descriptors restricted to $\cell'(f)$ do not behave as nicely as when restricted to $\cell(f)$: for example, the sorted Hilbert decomposition is affine on $\cell(f)$, but typically not on $\cell'(f)$.
\end{remark*}

We have the following key properties of cells.

\begin{lemma}
    \label{lemma:cells-one-parameter}
    The function $(\Rbb^n)^K \to (\Rbb^K)^n$ mapping $g$ to $(g_1, \dots, g_n)$ is a linear diffeomorphism.
    Given $f \in \Fil_n(K)$, this function restricts to a linear diffeomorphism $I_f : \cell(f) \cong \cell(f_1) \times \cdots \times \cell(f_n)$.
\end{lemma}


\begin{lemma}
    \label{lemma:cells-semialgebraic}
    The set $\Fil_n(K) \subseteq (\Rbb^n)^K$ is semilinear, and, $\cell(f) \subseteq (\Rbb^n)^K$ is semilinear for every $f \in \Fil_n(K)$.
\end{lemma}


\begin{corollary}
    \label{corollary:semialgebraic-map-from-fil}
    Let $\alpha : \Fil_n(K) \to \Rbb^D$
    be a function such that $\alpha|_{\cell(f)} : \cell(f) \to \Rbb^D$ is semilinear (resp.~semialgebraic) for every $f \in \Fil_n(K)$.
    Then, $\alpha$ is semilinear (resp.~semialgebraic).
    \qed
\end{corollary}

\subsubsection{Grid inclusions and geometry of cells}

\begin{definition}
    Let $\Gcal = \Gcal_1 \times \cdots \times \Gcal_n$.
    An \emph{aligned grid inclusion} of $\Gcal$ into $\Rbb^n$ is a monotonic and injective map $\Gcal \to \Rbb^n$ that is a product of monotonic and injective maps $\Gcal_k \to \Rbb$ for $1 \leq k \leq n$.
\end{definition}

As an example, for any filtration $f \in \Fil_n(K)$, the map $\iota_f : \grid_f \to \Rbb^n$ is an aligned grid inclusion.

The following construction allows us to identify the space of aligned grid inclusions of a fixed grid $\Gcal$ as a particularly simple subset of a Euclidean space.

\begin{construction}
    \label{construction:aligned-grid-inclusions-and-incl}
Given $m \in \Nbb$, define $\incl([m]) \subseteq \Rbb^{[m]}$
\[
    \incl([m]) \coloneqq \left\{f : [m] \to \Rbb \mid f(0) < \dots < f(m-1)\right\},
\]
and, for a grid $\Gcal = \Gcal_1 \times \cdots \times \Gcal_n$, let
\[
    \incl(\Gcal) \coloneqq \incl(\Gcal_1) \times \cdots \times \incl(\Gcal_n) \subseteq \Rbb^{[m_1]} \times \cdots \times \Rbb^{[m_n]}.
\]
Note that, if $\kappa : \Gcal \to \Rbb^n$ is an aligned grid inclusion, then $\kappa = \kappa_1 \times \cdots \times \kappa_n$, and $(\kappa_1, \dots, \kappa_n) \in \incl(\Gcal)$.
\end{construction}

Thanks to \cref{construction:aligned-grid-inclusions-and-incl}, we can, and do, identify aligned grid inclusions of $\Gcal$ into $\Rbb^n$ with $\incl(\Gcal)$.
In particular, if $f \in \Fil_n(K)$, we write $\iota_f \in \incl(\grid_f)$.


\begin{lemma}
    \label{lemma:incl-open-semilinear}
    For any grid $\Gcal$, the set $\incl(\Gcal) \subseteq \Rbb^{[m_1]} \times \cdots \times \Rbb^{[m_n]}$ is open, convex, and semilinear.
\end{lemma}

Our goal now is to prove that the map taking a filtration $f$ to its corresponding aligned grid inclusion restricts to a linear diffeomorphism between the cell of $f$ and the space of inclusions $\incl(\grid_f)$.

We start by formally defining this map.

By \cref{definition:cell}, if $g \in \cell(f)$, then $\grid_g = \grid_f$, so that $\iota_g$ is an aligned grid inclusion of $\grid_f$ into $\Rbb^n$, and thus $\iota_g \in \incl(\grid_f)$.
This allows us to define the following.

\begin{definition}
    \label{definition:iota}
    Let $f \in \Fil_n(K)$.
    Define the function $\itsgridincl : \cell(f) \to \incl(\ord_f)$ by mapping $g \in \cell(f)$ to $\iota_g : \grid_f = \grid_g \to \Rbb^n$, seen as an aligned grid inclusion (\cref{construction:aligned-grid-inclusions-and-incl}).
\end{definition}

\begin{definition}\label{definition:carrier}
    Let $f \in \Fil_1(K)$.
    A \emph{carrier} for $f$ is a function
    $C : \grid_f \to K$, such that $\ord_{f}(C(k)) = k$ for every $k \in \grid_f$.
    A \emph{carrier} for a multifiltration $f \in \Fil_n(K)$ consists of a family of functions $C = \{C_i : \grid_{f_i} \to K\}_{1 \leq i \leq n}$, with $C_i$ a carrier of $f_i$ for every $1 \leq i \leq n$. 
\end{definition}

Every filtration $f \in \Fil_n(K)$ admits a carrier, since $\ord_{f_i} : K \to \ord_{f_i}$ is surjective for all $1 \leq i \leq n$, by construction.

\begin{lemma}
    \label{lemma:carrier-iso}
    Let $f \in \Fil_n(K)$ and let $C$ be a carrier for~$f$.
    The function
    \[
        \prod_{i=1}^n\cell(f_i)
        \xrightarrow{\;\;\;\;C^*\;\;\;\;}
        \prod_{i=1}^n\incl(\grid_{f_i}) = \incl(\grid_f)
    \]
    given by mapping $(g_1, \dots g_n)$ to $(g_1 \circ C_1, \dots, g_n \circ C_n)$ is a linear diffeomorphism.
\end{lemma}

\begin{corollary}
    \label{corollary:cells-nice}
    Let $f \in \Fil_n(K)$ and let $C$ be a carrier for~$f$.
    Then, $\itsgridincl = C^* \circ I_f : \cell(f) \to \incl(\grid_f)$.
    Thus, $\itsgridincl : \cell(f) \to \incl(\grid_f)$ is a linear diffeomorphism between $\cell(f)$ and an open, convex, and semilinear subset of a Euclidean space.
\end{corollary}

Note, however, that the dimension of $\cell(f)$ depends on $f$ (more specifically, on $\grid_f$).

\subsection{Semilinear descriptors of multifiltrations}
\label{section:grid-determined}

In this section we prove \cref{theorem:grid-determined-implies-semialgebraic}, which gives sufficient conditions for an invariant to be semilinear as well as an explicit description of the gradient.
We instantiate this result to well known invariants in \cref{theorem:semilinearity-main-theorem}.

\begin{figure}
\[
    \begin{tikzpicture}
        \matrix (m) [matrix of math nodes,row sep=3em,column sep=3em,minimum width=2em,nodes={text height=1.75ex,text depth=0.25ex}]
        { \Fil_n(K)  &  \,  & \Rbb^D \\
          \cell(f) & \, & \incl(\grid_f) \\
          {\displaystyle \prod_{i=1}^n \cell(f_i)} & \, & {\displaystyle \prod_{i=1}^n \incl(\grid_{f_i})}\\
          };
        \path[line width=0.75pt, -{>[width=8pt]}]
        (m-1-1) edge [above] node {$\alpha$} (m-1-3)
        (m-2-1) edge [left] node {$\cong$} node [right] {$I_f$ {\scriptsize (Lem.~\ref{lemma:cells-one-parameter})}} (m-3-1)
        (m-2-1) edge [above] node {$\itsgridincl$ {\scriptsize (Def.~\ref{definition:iota})}} node [below] {$\cong$ {\scriptsize (Cor.~\ref{corollary:cells-nice})}} (m-2-3)
        (m-2-1) edge [right hook-{>[width=8pt]}] (m-1-1)
        (m-3-1) edge [below] node {$C^*$ {\scriptsize (Lem.~\ref{lemma:carrier-iso})}} node [above] {$\cong$ } (m-3-3)
        (m-3-3) edge [double equal sign distance,-] (m-2-3)
        (m-2-3) edge [left] node {$\push^\alpha_{\ord_f}$} node [right] {\scriptsize (Prop.~\ref{propositon:push-existence})} (m-1-3)
        ;
    \end{tikzpicture}
\]
\caption{Factoring a descriptor $\alpha$ through each cell allows us to deal with functions over linear, open subsets of a Euclidean space.}
\label{figure:key-diagram}
\end{figure}

\begin{definition}
Let $\Gcal$ be a grid.
A \emph{$\Gcal$-filtration} of a simplicial complex $K$ consists of a monotonic function $K \to \Gcal$.
A $\Gcal$-filtration is \emph{componentwise surjective} if $h_i : K \to \Gcal_i$ is surjective for each $1 \leq i \leq n$.
\end{definition}

The set of $\Gcal$-filtrations on $K$ is denoted by $\Fil_\Gcal(K)$ \label{grid-filts}.
As an example, $\ord_f \in \Fil_{\grid_f}(K)$ is a componentwise surjective $\grid_f$-filtration,
for every filtration $f \in \Fil_n(K)$.

For an interpretation of the next result, see \cref{remark:push-interpretation}.

\begin{proposition}
    \label{propositon:push-existence}
    Let $\alpha : \Fil_n(K) \to A$.
    For every grid $\Gcal$ and componentwise surjective $\Gcal$-filtration $h \in \Fil_\Gcal(K)$, there exists a unique function $\push^\alpha_h : \incl(\Gcal) \to A$ such that, for every $f \in \Fil_n(K)$ with $\grid_f = \Gcal$ and $\ord_f = h$ we have $\alpha|_{\cell(f)} = \push^\alpha_{\ord_f} \circ \itsgridincl : \cell(f) \to A$.
\end{proposition}

Graphically, there is a uniquely determined vertical map making the top rectangle of \cref{figure:key-diagram} commute; we now give an interpretation of this map.

\begin{remark}
    \label{remark:push-interpretation}
    Let $\alpha : \Fil_n(K) \to A$, and fix a grid $\Gcal$ and componentwise surjective $\Gcal$-filtration $h \in \Fil_\Gcal(K)$.
    The map $\push^\alpha_h : \incl(\Gcal) \to A$ is to be interpreted as the map that first computes the invariant $\alpha$ on the discrete, $\Gcal$-filtration $h$, and then pushes this computation along any given grid inclusion $\kappa \in \incl(\Gcal)$.
    Informally, \cref{propositon:push-existence} says that any invariant of multifiltrations is determined by an invariant of discrete filtrations (indexed by finite grids), and by a map $\push$ whose job is to push the invariant of a discrete filtration along grid inclusions.
\end{remark}

\begin{definition}
\label{def:semilinearly-determined}
  A function $\alpha : \Fil_n(K) \to \Rbb^D$ is \emph{semilinearly determined on grids} if, for each grid $\Gcal$ and componentwise surjective $\Gcal$-filtration $h \in \Fil_\Gcal(K)$, the map $\push^\alpha_h : \incl(\Gcal) \to \Rbb^D$ is semilinear.
\end{definition}

\begin{remark}
    \label{remark:definably-determined-on-grids}
More generally, one can define the notions of a function $\alpha$ being \emph{definably determined on grids} with respect to some o-minimal structure.
However, this level of generality is not required to handle multiparameter descriptors we have considered.
\end{remark}


The next theorem and proposition follow immediately from \cref{corollary:semialgebraic-map-from-fil,lemma:incl-open-semilinear,lemma:cells-one-parameter,corollary:cells-nice}.

\begin{theorem}
    \label{theorem:grid-determined-implies-semialgebraic}
    If $\alpha : \Fil_n(K) \to \Rbb^D$ is semilinearly determined on grids, then it is a semilinear function.\qed
\end{theorem}

\begin{proposition}
    \label{proposition:gradient}
    If $\alpha : \Fil_n(K) \to \Rbb^D$ is semilinearly determined on grids, then, for every $f \in \Fil_n(K)$, an element of the Clarke differential of the restriction $\alpha|_{\cell(f)} : \cell(f) \to \Rbb^D$ can be obtained by pulling back an element of the Clarke differential of $\push^\alpha_{\ord_f} : \incl(\grid_f) \to \Rbb^D$, which is a semilinear function with domain an open set of a Euclidean space, along the linear diffeomorphism $C^* \circ I_f$, for $C$ any carrier of $f$.\qed
\end{proposition}

\begin{remark}
    \label{remark:gradient-descent-implementation}
    The upshot of \cref{theorem:grid-determined-implies-semialgebraic,proposition:gradient} is that, in order to implement subgradient descent for an invariant $\alpha : \Fil_n(K) \to \Rbb^D$, it is enough to provide a map $\push_h^\alpha : \incl(\Gcal) \to \Rbb^D$, for each grid $\Gcal$ and $\Gcal$-filtration $h \in \Fil_\Gcal(K)$, for which subgradients are known.
    Recall that the interpretation of $\push_h^\alpha$ is in \cref{remark:push-interpretation}.
\end{remark}

%
%
%

\subsection{Applications}
\label{section:applications}

We now describe how to use \cref{theorem:satisfy-locally-lipschitz-condition,corollary:Hilbert-and-rank-semilinear} to obtain definable and locally Lipschitz objective functions based on the Hilbert decomposition signed measure; analogous constructions work for the rank decomposition signed measure.
Similar pipelines can be constructed using the multiparameter persistence landscape, thanks to its semilinearity \cref{proposition:landscape-semilinear} and stability \cite{Vipond2020}.
Recall that $\OT$ is the optimal transport distance (\ref{section:POT}).

\begin{example}
    \label{example:distance-to-a-measure}
    Let $\dsm_k(\Rbb^n) \subseteq \dsm(\Rbb^n)$ denote the subset of signed measure of total mass $k \in \Zbb$ (that is, discrete measures such that the number of positive point masses minus the number of negative point masses is $k$)\footnote{One can see that, for every $f \in \Fil_n(K)$, we have $\mu^{\Hil}_{H_i(f)} \in \dsm_k(\Rbb^n)$ for $k = \dim(H_i(K))$, which only depends on $K$ and not on the specific filtration $f$.}.
    Fix $\nu \in \dsm_k(\Rbb^n)$.
    Then, the function $E' : \dsm_k(\Rbb^n) \to \Rbb$ given by mapping $\mu$ to $\OT(\mu, \nu) \in \Rbb$ is $1$-Lipschitz.
    Since the computation of $\OT(\mu,\nu)$ reduces to computing a minimum (over all permutations of the point masses of $\mu^+ + \nu^-$ and $\nu^+ + \mu^-$) of sums of distances, it is easy to find a representative $E : \Rbb^\indexingset \to \Rbb$ as in \cref{theorem:satisfy-locally-lipschitz-condition}.
    And, this map $E'$ is in fact semialgebraic.
    Combining this with \cref{corollary:Hilbert-and-rank-semilinear}, we conclude that we can use the distance to a fixed signed measure $\nu$, i.e., the following map,
    as a topological objective satisfying the assumptions of Thm.~\ref{theorem:main-results}:
    \[
        \Fil_n(K) \xrightarrow{\;\;\;\;} \Rbb
        \;\;\text{given by}\;\;
        f \xmapsto{\;\;\;\;} \OT( \mu^\Hil_{H_i(f)}, \nu).
    \]
\end{example}

\begin{example}
    \label{example:integration-function}
    Let $\psi : \Rbb^n \to \Rbb$ be a Lipschitz function.
    We get a Lipschitz function $\dsm(\Rbb^n) \to \Rbb$
    by mapping $\mu$ to $\int \psi\, \dsf \mu \in \Rbb$.
    Since the computation of $\int \psi\, \dsf \mu$ reduces to a finite sum indexed by the point masses of $\mu$ and weighted by $\psi$, it is easy to find a representative
    $E : \Rbb^\indexingset \to \Rbb$ as in \cref{theorem:satisfy-locally-lipschitz-condition}.
    If the function $\psi$ is definable over some o-minimal structure, then so is $E'$.
    Combining this with \cref{corollary:Hilbert-and-rank-semilinear}, we conclude that we can use integration with respect to $\psi$, i.e., the following map, as a topological objective which satisfies assumptions of \cref{theorem:main-results}:
    \[
        \Fil_n(K) \xrightarrow{\;\;\;\;} \Rbb
        \;\;\text{given by}\;\;
        f \xmapsto{\;\;\;\;} \int \psi \, \dsf\mu^\Hil_{H_i(f)}.
    \]
\end{example}

\section{Numerical experiments}
\label{sec:experiments}


We first illustrate our framework on a toy point cloud optimization experiment inspired from~\cite{Carriere2021a}. 
We then show experiments related more specifically to ML applications with synthetic and real-world data (see also \cref{section:experiments_appendix}, for details).
We use the implementation in \cite{multipers}.


\noindent\textbf{Point cloud optimization.}
This is not an application per se, but rather an illustrative example with the following two main goals:
(1) To demonstrate, with an experiment that can be easily assessed visually, that our pipeline works as expected, allowing the user to perform multiparameter homological optimization.
(2) To demonstrate that multiparameter homology optimization enables the user to optimize with respect to topological structures that can only be captured by considering the relationship between two or more parameters (in this case metric and density structure).
We build upon the point cloud example of \cite{Carriere2021a}.

Given a finite set of points $X$ sampled uniformly from the unit square in~$\R^2$ (\cref{fig1:pt-cloud}), the objective function aims to maximize the distance between the Hilbert decomposition signed measure of the optimized point cloud and the zero signed measure. Intuitively, this makes the topological content of the optimized point cloud as non-trivial as possible.
This is in direct analogy with the loss used in \cite{Carriere2021a}, which maximizes the optimal transport distance (referred to as Wasserstein distance there) to the zero persistence diagram.
In the one-parameter setting this is achieved by increasing the quantity and ``size'' (persistence) of the ``holes'' (cycles), while, in our two-parameter setting, the quantity and size of the holes, as well as the density of the constituent points, play a role in the objective function.

The Vietoris--Rips filtration $\VR$ of~$X$---see Definition~5.2 of \cite{oudot}---is used for this purpose, via the maximization of the sum of the lengths of the bars in its barcode\footnote{I.e., one-parameter rank decomposition measure (App.~\ref{section:rank-dec-signed-measure}).}
in homology degree~1. Up to a constant factor, this can be rewritten as the $\OT$ distance between the zero measure and the Hilbert decomposition signed measure of~$H_1(\VR)$:
\begin{equation}\label{eq:loss_pd}
\OT\left(\,\mu^\Hil_{H_1(\VR)}(X)\,,\, 0\right).
\end{equation}
\begin{figure*}[th]
    \centering
    \begin{subfigure}{0.23\textwidth}
        \includegraphics[width=\linewidth]{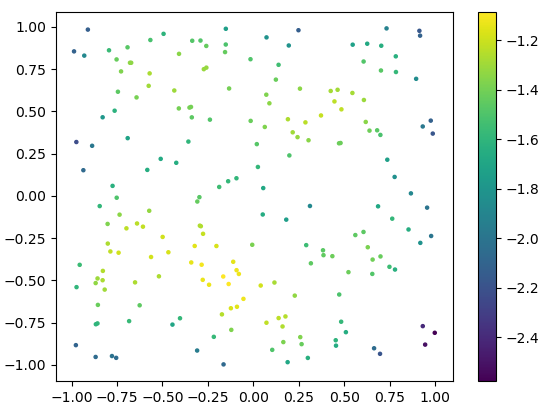}
        \caption{Point cloud $X$}
        \label{fig1:pt-cloud}
    \end{subfigure}
    \hfill
    \begin{subfigure}{0.23\textwidth}
        \includegraphics[width=\linewidth]{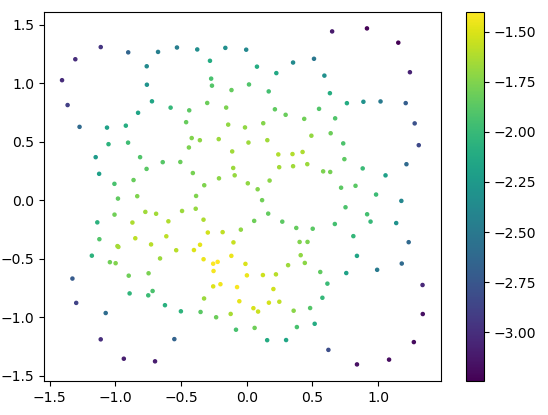}
        \caption{$X$ optimized with~\eqref{eqn:norm_loss}}
        \label{fig1:OT2}
    \end{subfigure}
    \hfill
    \begin{subfigure}{0.23\textwidth}
        \includegraphics[width=\linewidth]{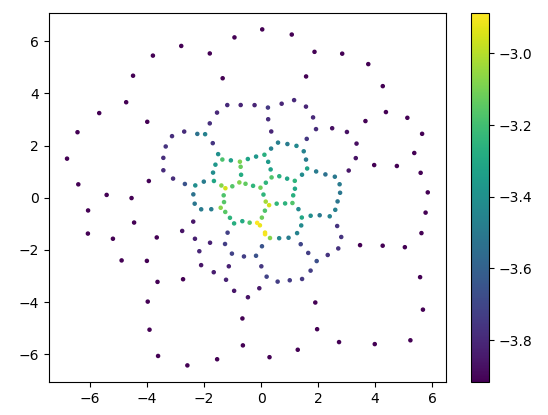}
        \caption{$X$ optimized with~\eqref{eq:loss_pd}}
        \label{fig1:OT1}
    \end{subfigure}
    \hfill
    \begin{subfigure}{0.23\textwidth}
        \includegraphics[width=\linewidth]{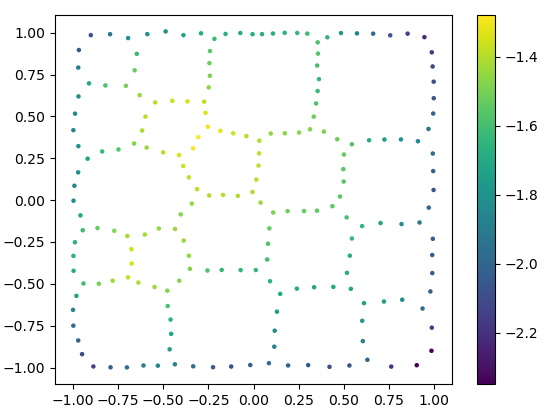}
        \caption{$X$ optimized w/~\eqref{eq:loss_pd} + regul.}
        \label{fig1:OT1+regul}
\end{subfigure}

    \begin{subfigure}{0.23\textwidth}
        \includegraphics[width=\linewidth]{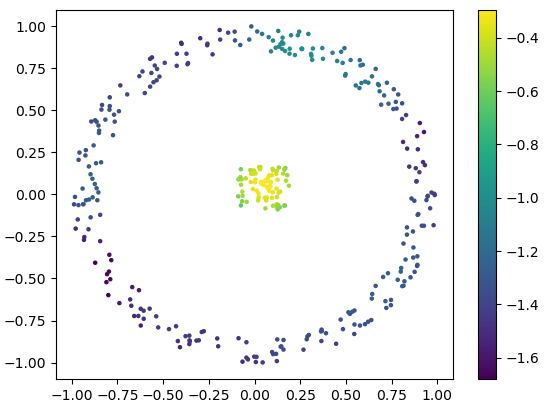}
        \caption{Point cloud $Y$}
        \label{fig2:pt-cloud}
\end{subfigure}
    \hfill
    \begin{subfigure}{0.23\textwidth}
        \includegraphics[width=\linewidth]{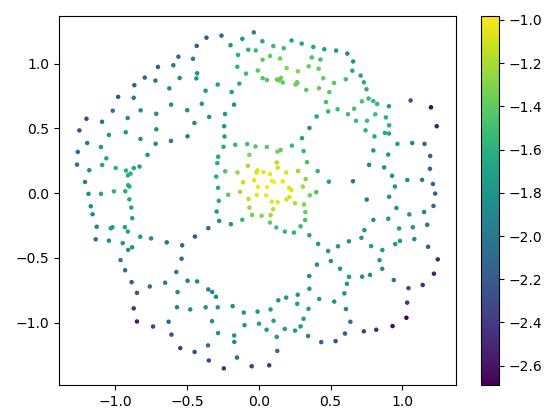}
        \caption{$Y$ optimized with~\eqref{eqn:norm_loss}}
        \label{fig2:OT2}
    \end{subfigure}
    \hfill
    \begin{subfigure}{0.23\textwidth}
        \includegraphics[width=\linewidth]{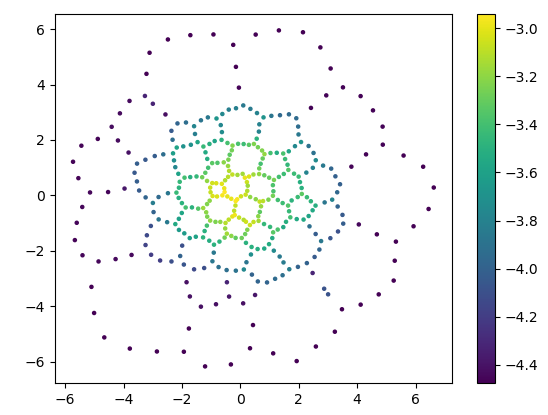}
        \caption{$Y$ optimized with~\eqref{eq:loss_pd}}
        \label{fig2:OT1}
    \end{subfigure}
    \hfill
    \begin{subfigure}{0.23\textwidth}
        \includegraphics[width=\linewidth]{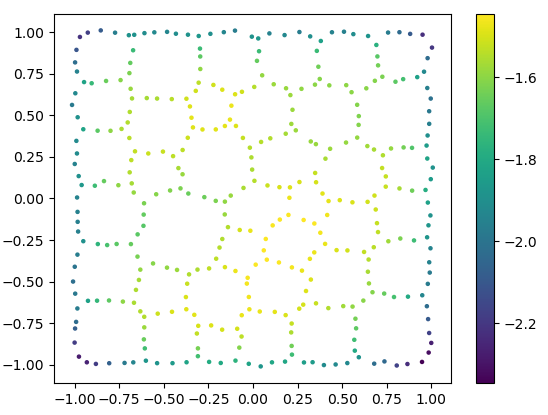}
        \caption{$Y$ optimized w/~\eqref{eq:loss_pd} + regul.}
        \label{fig2:OT1+regul}
    \end{subfigure}
    \caption{Optimizing the holes of point clouds.
    The colors indicate the $\log$-values of the density estimator.}
    \label{fig:one_circle_uni_noise}
\end{figure*}
While this approach does succeed in creating many cycles, it has two major drawbacks.
(a) As observed in~\cite{Carriere2021a} and illustrated in \cref{fig1:OT1}, the topological loss \eqref{eq:loss_pd} scales with rescaling of the points of~$X$, so loss~\eqref{eq:loss_pd} alone causes the points to diverge to infinity. To mitigate this effect, another purely geometric loss is added in~\cite{Carriere2021a}, which restricts the points to the unit square. (b) There is no control over the actual sizes of the obtained holes, which end up being arbitrary---see \cref{fig1:OT1+regul}.
Here we address these two issues at the same time by using the 
function-Rips bifiltration~$(\VR,\hat f)$ of~$X$---see Definition 5.1 of \cite{botnan-lesnick}---for some data-dependent kernel density estimator~$\hat f$, and by maximizing the analog of loss~\eqref{eq:loss_pd} in this two-parameter setting:
\begin{equation}\label{eqn:norm_loss}
\OT\left(\,\mu^\Hil_{H_1(\VR, \hat f)}(X)\,,\, 0\right).
\end{equation}
Note that the density estimator~$\hat f$ is recomputed after each epoch of the optimization.
As can be seen from the result (\cref{fig1:OT2}), adding in the density parameter prevents the points from diverging to infinity, the intuitive reason being that the density scales down with rescalings of the points of~$X$. For the same reason, the density parameter exerts some control over the sizes of the cycles produced, inducing larger cycles in lower-density areas.

Another distinguishing feature of the two-parameter metric- and density-aware topological loss is that it preserves some of the characteristics of the density of the initial point cloud. To see this, consider a dataset comprised primarily of points sampled from a large circle, with some extra noise of higher density located around the center (\cref{fig2:pt-cloud}). Using the function-Rips bifiltration and maximizing~\eqref{eqn:norm_loss} creates cycles from the noise points in the center, while still maintaining the initial large circle (\cref{fig2:OT2}). In contrast, one-parameter VR optimization treats all points equally, regardless of their density, and therefore destroys the large circle (Figures~\ref{fig2:OT1} and~\ref{fig2:OT1+regul}).
A precise explanation for this phenomenon remains to be determined; we suspect that this behavior is similar to that of mean-shift, in that creating many, or large, holes in high-density areas requires many points in those areas.

\noindent\textbf{Topological dimension reduction.}
In this experiment, we present an application of multiparameter homological loss differentiation to dimension reduction with autoencoders.
Regularizing standard autoencoders with topological losses that constrain the latent space to have the same PH as the input space was one of the first applications of PH-based optimization to appear in the literature~\cite{Moor2020, Doraiswamy2020, Carriere2021a, Vandaele2022}.
Common topological losses for autoencoders have a regularization term involving a distance between the barcodes of one-parameter $\VR$ filtrations:
\begin{equation}\label{eq:autoencoder_loss}
    \mathcal L(\theta):=\OT\left(\,\mu^\rk_{H_*(\VR)}(\tilde X(\theta))\,,\, \mu^\rk_{H_*(\VR)}(X)\,\right),
\end{equation}
where $\theta$ denotes the parameters of the autoencoder, and $X$ and $\tilde X$ denote the input and the learned latent spaces, respectively.
This loss does not take density into account, making the quality of the learned latent spaces susceptible to the presence of even mild noise.
We compare this to using the following loss, based on the Hilbert decomposition signed measure and a function-Rips two-parameter filtration,
\begin{equation*}\label{eq:mp_autoencoder_loss}
\mathcal L(\theta):=\OT\left(\,\mu^\Hil_{H_*(\VR,\hat f)}(\tilde X(\theta))\,,\, \mu^\Hil_{H_*(\VR,\hat f)}(X)\,\right),
\end{equation*}
where $\widehat{f}$ is a data-dependent density estimator.
This topological loss is an instance of~\cref{example:distance-to-a-measure}.

The point cloud data consist of two interlaced circles with background noise embedded in $\R^9$ (\cref{figure:autoencoder-data}), similar to the data used in~\cite{Carriere2021a}. 
See \autoref{fig:autoencoder}, which shows that the proposed multiparameter topological regularization outperforms both no topological regularization, and the one-parameter regularization of \cref{eq:autoencoder_loss}.
See \cref{section:details-autoencoder} for details about this experiment.

\begin{figure}[th]
\centering
\includegraphics[width=2.65cm]{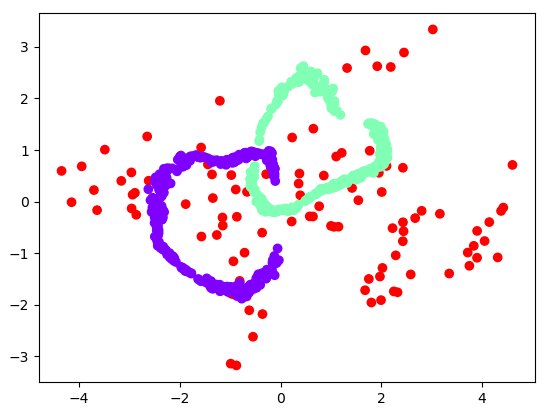}
\includegraphics[width=2.65cm]{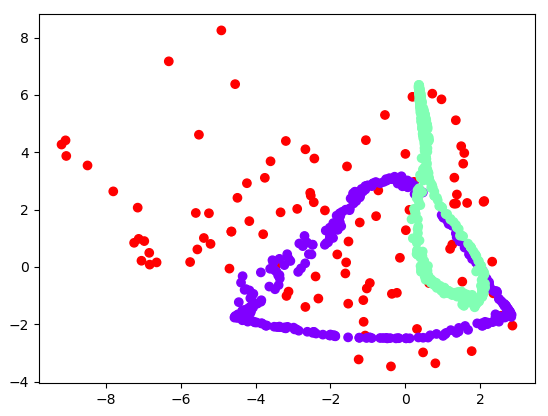}
\includegraphics[width=2.65cm]{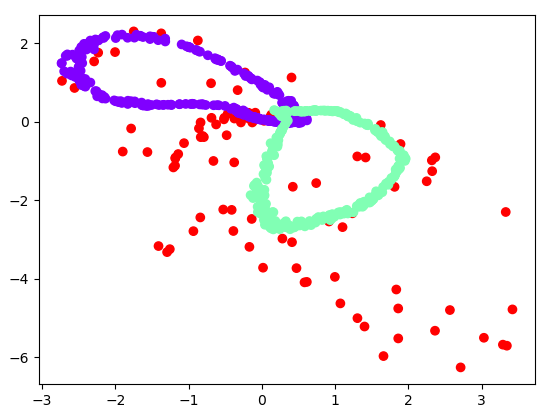}

\includegraphics[width=2.65cm]{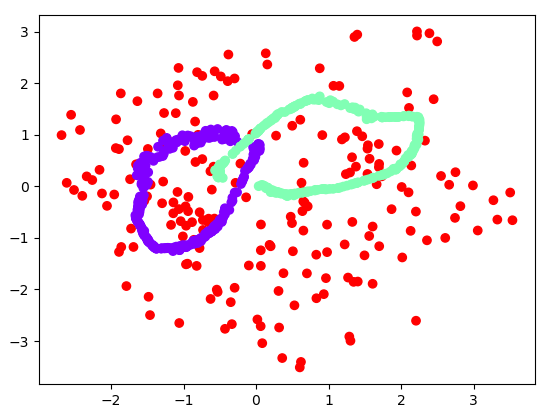}
\includegraphics[width=2.65cm]{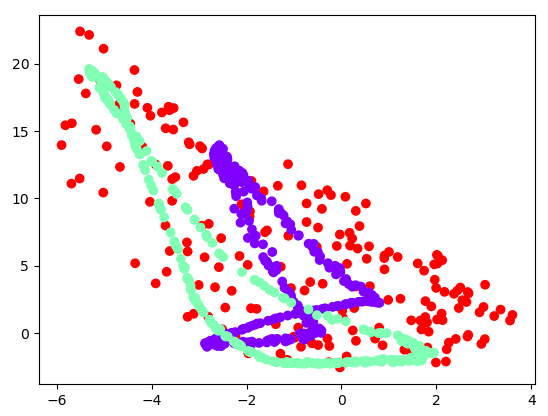}
\includegraphics[width=2.65cm]{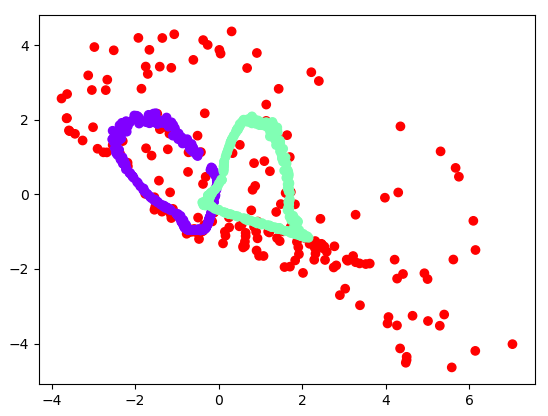}
\caption{Rows correspond to datasets (top has less background noise), and columns correspond to no topological regularization, one-parameter regularization, and multiparameter regularization.
One-parameter regularization is very susceptible to noise points, while no topological regularization can fail to close the circles and preserve topology.
Interestingly, no topological regularization behaves better with many noise points, possibly due to the metric loss having more distances to work with.
Multiparameter topological regularization ensures the preservation of topology in both cases.}
\label{fig:autoencoder}
\end{figure}

\smallskip \noindent\textbf{Graph classification.}
In this experiment, we illustrate the use of multiparameter topological optimization for deep learning on graph data, since graph classification is known to be a useful application of topology-based methods~\cite{Zhao2019, hofer-2, Carriere2020, Zhao2020, Horn2022, Loiseaux2023}.
We use the setup developed in~\cite{Horn2022}, in which node and edge attributes are learned at the same time that a graph neural network classifier is trained, with the attributes being computed using topological descriptors.
The authors of \cite{Horn2022} study the performance of using $k$ one-parameter descriptors, using $k$ different node attributes.
We compare the use of $k$, independent, one-parameter topological descriptors, against the use of a single $k$-parameter topological descriptor.
Since attributes are learned, the descriptors used must be differentiable, and this is where the application of our theory lies.
Details are in \cref{subsec:details-graph}.

\cref{table:graph-data-10} shows classification performances of a graph neural network architectures
complemented with either no topology (*-NOTOP), one-parameter PH as in~\cite{Horn2022} (*-OPTOP), and multiparameter PH (*-MPTOP).

Our goal is \emph{not} to obtain state-of-the-art results, but rather show a positive improvement when using multiparameter homological descriptors instead of one-parameter ones.

\begin{table}[th]
\centering
    {\tiny
    \begin{tabular}{|l|l|l|l|l|}
        \hline
        Model               & ENZYMES                      & IMDB-B                      & IMDB-M                      & MUTAG                       \\ \hline
        \hline
        GCN-NOTOP           & 30.3$\pm$8.1                 & 73.2$\pm$6.4                & 44.9$\pm$7.6                & \textbf{87.2$\pm$5.6}       \\
        GCN-OPTOP           & 28.8$\pm$7.5                 & 75.2$\pm$5.6                & {\ul \textbf{51.2$\pm$4.4}} & 84.1$\pm$8.9                \\
        GCN-MPTOP           & {\ul \textbf{39.0$\pm$10.1}} & {\ul \textbf{78.4$\pm$5.1}} & 51.1$\pm$3.5                & {\ul 85.1$\pm$7.7}          \\ \hline
        \hline
        GIN-NOTOP           & \textbf{47.0$\pm$12.9}       & 71.2$\pm$5.4                & 47.1$\pm$2.9                & 87.2$\pm$8.0                \\
        GIN-OPTOP           & 45.3$\pm$11.8                & {\ul \textbf{75.0$\pm$2.7}} & 47.5$\pm$5.0                & {\ul \textbf{88.3$\pm$8.9}} \\
        GIN-MPTOP           & {\ul 46.5$\pm$11.2}          & 71.3$\pm$5.1                & {\ul \textbf{48.5$\pm$4.2}} & 87.2$\pm$6.1                \\ \hline
        \hline
        GraphResNet-NOTOP   & 42.8$\pm$11.0                & \textbf{75.3$\pm$5.3}       & 49.4$\pm$4.3                & 88.8$\pm$5.2                \\
        GraphResNet-OPTOP   & 39.5$\pm$12.2                & 68.1$\pm$8.2                & 40.7$\pm$3.5                & 87.8$\pm$4.3                \\
        GraphResNet-MPTOP   & {\ul \textbf{44.3$\pm$9.8}}  & {\ul 69.4$\pm$5.8}          & {\ul \textbf{50.1$\pm$4.4}} & {\ul \textbf{89.3$\pm$6.1}} \\ \hline
        \hline
        GraphDenseNet-NOTOP & 43.2$\pm$10.4                & 50.3$\pm$5.9                & 33.1$\pm$2.7                & 88.8$\pm$5.2                \\
        GraphDenseNet-OPTOP & 47.3$\pm$12.3                & 50.0$\pm$7.1                & 32.7$\pm$4.2                & 86.2$\pm$8.3                \\
        GraphDenseNet-MPTOP & {\ul \textbf{48.0$\pm$11.4}} & {\ul \textbf{52.2$\pm$7.7}} & {\ul \textbf{34.1$\pm$3.1}} & {\ul \textbf{92.6$\pm$5.1}} \\ \hline
        \end{tabular}
    }
    \caption{Bold denotes best score and underline best between one- and multiparameter topology.
        As one can see, it is usually better to use differentiable multiparameter topological descriptors.}
  \label{table:graph-data-10}
\end{table}

\section{Conclusions}

In order to address the fact that there is no unique homological descriptor of multifiltrations, we introduce a theoretical framework for studying the differentiability and gradient descent convergence of general topology-based losses of multiparameter filtrations.
Our main theoretical results give conditions ensuring differentiability and convergence of the stochastic subgradient method, and show that this applies to topological descriptors of very different flavors,
demonstrating the flexibility of our approach.
Our numerical experiments show positive improvements with respect to one-parameter homological optimization.

\smallskip
\noindent\textbf{Limitations and future work.}
An inherent limitation of multiparameter persistent homology is that it is generally computationally demanding.
This can be dealt with in practice by subsampling or other sparsification techniques \cite{alonso}.
To complement this, one could adapt, to the multiparameter case, the optimization with big steps of \cite{nigmetov}, to speed up optimization.

It remains a possibility that our theory does not apply to topological descriptors to be introduced in the future; we believe that suitable generalizations (e.g., descriptors that are definably determined on grids, as in \cref{remark:definably-determined-on-grids}), will at the very least encompass a large class of descriptors.

On the applications side, future work includes regularizing NNs \cite{chen} with multiparameter persistence.

\newpage
\section*{Acknowledgements}
L.S.~was partially supported by EPSRC grant ``New Approaches to Data Science: Application Driven Topological Data Analysis'', EP/R018472/1.
For the purpose of Open Access, the authors have applied a CC BY public copyright licence to any Author Accepted Manuscript (AAM) version arising from this submission.
D.L.~was supported by ANR grant ``3IA Côte d'Azur'', ANR-19-P3IA-0002. 
M.C.~was partially supported by ANR grant ``TopModel'', ANR-23-CE23-0014.

\section*{Impact statement}
This paper presents work whose goal is to advance the field of Machine Learning. There are many potential societal consequences of our work, none which we feel must be specifically highlighted here.

\bibliography{main}
\bibliographystyle{icml2024}

\newpage
\appendix
\onecolumn

\section{More background}
\label{section:more-background}

\subsection{Partial optimal transport}
\label{section:POT}

\subsubsection{Definition}
\label{section:pot-definition}

We define optimal transport between signed measures on a metric spaces, and partial optimal transport between signed measures on a metric pair.
This allows one to model the scenario in which one considers signed measures on a space $M$ and optimal transport plans that are allowed to ``throw mass away'' in a certain subset $Z \subseteq M$.
Partial optimal transport is relevant in topological data analysis \cite{divol-lacombe} and in optimal transport problems in general \cite{figalli-gigli}.
For more context, see \cref{section:context-POT}.

Let $M$ be an \emph{extended metric space}, that is, a set together with a function $d_M : M \times M \to [0,\infty]$ that is symmetric, which satisfies the triangle inequality, and such that $d_M(x,y) = 0$ if and only if $x=y \in M$.

The set $\dm(M)$ of discrete measures on $M$ admits an extended metric, called the \emph{optimal transport distance} and denoted $\OT$, which we now define.
Given $\mu, \nu \in \dm(M)$, let
\[
    \assignment(\mu,\nu) \coloneqq 
    \left\{
        \left(I \text{ fin.~set}\,,\,\, \substack{\beta\, :\, I \to M\\
        \gamma\, :\, I \to M}\right) : \substack{\mu\; =\; \sum_{i \in I} \delta_{\beta(i)}\\
        \nu\; =\; \sum_{i \in I} \delta_{\gamma(i)}}
    \right\},
\]
and for every $a = (I,\beta,\gamma) \in \assignment(\mu,\nu)$, define $\cost(a) \coloneqq \sum_{i \in I} d_M(\beta(i), \gamma(i))$.
Then, the set $\dm(M)$ together with the function
\[
    \OT(\mu,\nu) \coloneqq \inf \{ \cost(a) : a \in \assignment(\mu,\nu)\},
\]
with the convention that the infimum over an empty set is $+\infty$,
is an extended metric space.

The optimal transport distance extends to $\dsm(M)$, the set of discrete signed measures on $M$, by reduction to the positive case: define $\OT(\mu,\nu) \coloneqq \OT(\mu^+ + \nu^-, \nu^+ + \mu^-)$, where the signs $+$ and $-$ in superscripts indicate the positive and negative parts of the signed measures according to their Jordan decomposition.

Then, discrete signed measures on a metric pair are defined as follows.

\begin{definition}
    \label{definition:discrete-signed-measures-metric-pair}
    If $Z \subseteq M$ is any subset, the set of \emph{discrete signed measures} on the pair $(M,Z)$ is $\dsm(M,Z) \coloneqq \dsm(M \setminus Z)$. 
\end{definition}

It may seem from \cref{definition:discrete-signed-measures-metric-pair} that the subset $Z$ is hardly playing any role; this is not so, since the partial optimal transport distance (\cref{definition:pot}) uses $Z$ in a crucial way.

One can extend the optimal transport distance on discrete signed measures to the partial optimal transport distance on the set $\dsm(M,Z)$ of discrete signed measures on $(M,Z)$ that, informally, allows one to ``throw mass away'' in the subset $Z$.
This is defined formally as follows.

\begin{definition}
    \label{definition:pot}
Let $\mu,\nu \in \dsm(M,Z)$, their \emph{partial optimal transport distance} is
\[
    \OT(\mu,\nu) \coloneqq
    \inf\left\{\,\OT(\overline{\mu},\overline{\nu})\; :\;
            \overline{\mu}\,,\;\overline{\nu} \,\in\, \dsm(M)\,,\;\; \mu\; =\; \overline{\mu}|_{M\setminus Z}, \;\; \nu\; =\; \overline{\nu}|_{M\setminus Z}\,
        \right\}.
\]
\end{definition}

Then, the set $\dsm(M,Z)$ endowed with the partial optimal transport distance is an extended metric space (see \cref{section:context-POT}).

We use the same notation for the optimal transport distance and the partial optimal transport distance, as these are defined over two different spaces.

\subsubsection{Context}
\label{section:context-POT}
Note that, depending on the context, optimal transport distances are named after different people, including Kantorovich, Monge, Rubinstein, and Wasserstein.

Since in this paper we deal with discrete measures, our optimal transport distances take a particularly simple form, which reduces its computation to an assignment problem; we refer the reader to \cite{bubenik-elchesen} for a detailed exposition---from the lens of topological data analysis---of the concepts introduced here; in the cited paper, the measures we consider are called virtual persistence diagrams (Definition~4.10 of \cite{bubenik-elchesen}).
Then, Corollary~4.9 of \cite{bubenik-elchesen} implies that our partial optimal transport distance (\cref{definition:pot}) does in fact endow $\dsm(M,Z)$ with an extended metric.

Partial optimal transport between signed measures is not as develop as its positive (i.e., unsigned) counterpart, but note that Proposition~5.19 of \cite{bubenik-elchesen} gives very general conditions under which their notion of partial optimal transport for positive measures (which, in the case of discrete measures restricts to ours) coincides with that of \cite{figalli-gigli}.

\subsection{The rank decomposition signed measure}
\label{section:rank-dec-signed-measure}

The rank decomposition signed measure is a generalization of the well-known one-parameter barcode (also known as persistence diagram); and instead of being a measure on the space of intervals (or bars) in $\Rbb$, it is a signed measure on the space of higher dimensional bars, which we now define.
If $n \geq 1 \in \Nbb$, let
\[
    \bars_n = \{(r,s) \in \left(\Rbb^n \cup \{\infty\}\right)^2 : r \leq s\} \subseteq \left(\Rbb^n \cup \{\infty\}\right)^2.
\]
As in the one-parameter case, short bars are bars close to the diagonal, defined as $\Delta \coloneqq \{(r,r)\} \subseteq \bars_n$.

\begin{definition}
    \label{definition:rank-signed-measure}
    Let $f \in \Fil_n(K)$ and let $i \in \Nbb$.
    The $i$th \emph{rank decomposition signed measure} is the unique signed measure $\mu_{H_i(f)}^\rk \in \dsm\left(\bars_n, \Delta\right)$ such that, for all $r \leq s \in \Rbb^n \cup \{\infty\}$,
    \[
        \rk\left(
            H_i(f)(r) \xrightarrow{\phi^f_{r,s}} H_i(f)(s)
        \right)
        = \mu^\rk_{H_i(f)}\left( C(r,s) \right),
    \]
    where
    $C(r,s) =
    \left\{ (r',s') \in \bars_n : r' \leq r, s' \nleq s \right\}$.
\end{definition}
Existence and uniqueness follows from Corollary~5.6 of \cite{botnan-2}.

In order to make $\bars_n$ into an (extended) metric space, we endow $\Rbb^n$ and $(\Rbb^n \cup \{\infty\})^2$ with the extended distance induced by $\|-\|_\infty$, with the convention that $|\infty-x| = |x-\infty| = \infty$ if $x\neq \infty$ and $|\infty - \infty| = 0$.

In particular, in the one-parameter case, $\dsm(\bars_1,\Delta)$ is the usual set of one-parameter persistence barcodes, and the descriptor $\mu^\rk$ is the one-parameter barcode.

\subsection{Isomorphism invariance}
\label{section:iso-invariance}

Let $K$ and $K'$ be simplicial complexes with underlying sets $X$ and $X'$, respectively.
An \emph{isomorphism} between $K$ and $K'$ consists of a bijection $\psi : X \to X'$ such that $\sigma = \{x_0, \dots, x_i\} \in K$ if and only if $\psi(\sigma) \coloneqq \{\psi(x_0), \dots, \psi(x_i)\} \in K'$.

\begin{example}
    Two graphs are isomorphic as graphs if and only if they are isomorphic as simplicial complexes.
\end{example}

An \emph{isomorphism} between filtrations $f \in \Fil_n(K)$ and $g \in \Fil_n(K')$ consists of an isomorphism of simplicial complexes $\psi : K \to K'$ such that $g \circ \psi = f$.

\begin{definition}
    \label{definition:iso-invariance}
    A family of descriptors $\{\alpha^K : \Fil_n(K) \to A\}_K$ indexed by all finite simplicial complexes $K$ is \emph{isomorphism invariant} if $\alpha(f) = \alpha(g) \in A$ whenever $f$ and $g$ are isomorphic filtrations.
\end{definition}

Since homology is isomorphism invariant, it follows that all descriptors that only rely on homology are also isomorphism invariants.
In particular, we have the following.

\begin{proposition}
    The Hilbert function, rank invariant, Hilbert decomposition signed measure, and rank decomposition signed measure are isomorphism invariant. \qed
\end{proposition}

\section{Convergence of the stochastic subgradient method with topological descriptors}

\subsection{Details about the stochastic subgradient method}

\begin{assumption}
    \label{assumption:davis}
    \begin{enumerate}[(a)]
        \item
        The sequence $\{a_k\}_{k \in \Nbb}$ is such that $a_k \geq 0$ for all $k \in \Nbb$, $\sum_{k} a_k = \infty$, and $\sum_k a_k^2 < \infty$.
        \item Almost surely, $\sup_k \|x_k\| < \infty$.
        \item
        The sequence of random variables $\{\zeta_k\}_{k \in \Nbb}$ has the property that there exists a function $p : \Rbb^d \to [0,\infty)$, bounded on bounded sets, such that
        \[
            \Ebb[\zeta_k | \Fcal_k] = 0 \text{ and } \Ebb[\|\zeta_k\|^2 | \Fcal_k] < p(x_k),
        \]
        almost surely, where $\Fcal_k = \sigma(x_j, y_j, \zeta_j)_{j < k}$ is the increasing sequence of $\sigma$-algebras generated by $x$, $y$, and $\zeta$ up to $k$.
    \end{enumerate}
\end{assumption}

We comment on how to satisfy the \cref{assumption:davis} in \cref{remark:satisfying-davis}

The following is a simple adaptation of a main result from \cite{davis-drusvyatskiy-kakade-lee}.

\begin{proposition}
    \label{theorem:convergence-theorem}
    Assume given a simplicial complex $K$, a number of parameters $n \geq 1 \in \Nbb$, and functions
    \[
        \Rbb^d \xrightarrow{\;\;\;\Phi\;\;\;} \Fil_n(K) \xrightarrow{\;\;\;\alpha\;\;\;} \Rbb^D \xrightarrow{\;\;\;E\;\;\;} \Rbb.
    \]
    Assume the following:
    \begin{enumerate}[leftmargin=*]
        \item The loss function $\Lcal \coloneqq \Phi \circ E \circ \alpha : \Rbb^d \to \Rbb$ is locally Lipschitz.
        \item The function $\alpha$ is semilinear.
        \item The functions $\Phi$ and $E$ are definable over a common o-minimal structure.
    \end{enumerate}
    Under \cref{assumption:davis},
    almost surely, every limit point of the iterates of the stochastic subgradient method is critical for the loss $\Lcal$, and the sequence of loss values converges.
\end{proposition}
\begin{proof}
    We use Corollary~5.9, part (stochastic subgradient method), of \cite{davis-drusvyatskiy-kakade-lee}.
    In the notation of that result, we are considering $f = \Lcal = E \circ \alpha \circ \Phi$, which is locally Lipschitz by assumption (1).
    By assumption (3), the functions $\Phi$ and $E$ are definable over a common o-minimal structure $\Ocal$, and by assumption (2), it follows that $\Lcal = E \circ \alpha \circ \Phi$ is definable over $\Ocal$ (since semilinear functions are definable over all o-minimal structures \cite{vandenDries}).
    Since functions definable over an o-minimal structure are Whitney $C^p$-stratifiable for any $p \geq 1 \in \Nbb$ \cite{vandenDries-Miller}, it follows that $f = \Lcal$ is Whitney $C^d$-stratifiable.
%
\end{proof}

\subsection{Proof of \cref{theorem:main-results}}

The statements not involving the optimization pipeline $E \circ \alpha \circ \Phi$ are the contents of \cref{theorem:grid-determined-implies-semialgebraic,proposition:gradient}.
For the convergence of the stochastic subgradient method, we use \cref{theorem:convergence-theorem}, and to satisfy condition (2) we use and \cref{theorem:grid-determined-implies-semialgebraic} and the fact that $\alpha$ is semilinearly determined on grids, by assumption.

\subsection{On the assumptions of \cref{theorem:main-results} and possible extensions}

We comment on possible uses and generalizations of \cref{theorem:convergence-theorem} and thus of \cref{theorem:main-results}.

\begin{remark}
    \label{remark:generalizations-of-main-theorem}
    The optimization pipeline considered in \cref{theorem:main-results} typically arises in ML contexts where PH is used as the initial data featurization step. For instance, when the input data are filtered point clouds to be optimized, the map $\Phi$ can build the corresponding function-Rips bifiltration, while the map~$E$ can be the loss associated to some neural network architecture into which the vectorized invariant is plugged.
    
    Beyond that, \cref{theorem:main-results} can easily be adapted to encompass a larger variety of use cases, using stronger results of \cite{davis-drusvyatskiy-kakade-lee}, such as those in Section~6 of \cite{davis-drusvyatskiy-kakade-lee}.
    For instance, one can consider composites $\Lcal \circ F$, where $F\colon \Rbb^p \to \Rbb^d$ is both locally Lipschitz and definable over the same o-minimal structure as~$\Lcal$ (for instance, $F$ could be some neural network), which enables for instance the use of PH as an intermediate step (as opposed to the initial step) in the learning pipeline, with back-propagation of the gradients via the chain rule.
    Alternatively, one can consider sums $\Lcal + \Lcal'$ where $\Lcal'$ is another objective function (possibly non-topological) that is both locally Lipschitz and definable over the same o-minimal structure as~$\Lcal$, which enables for instance the use of PH as a topological regularizer.
    Finally, one can also restrict the domain of $\Phi$, also as regularization.
\end{remark}

We now comment on the different conditions of \cref{theorem:convergence-theorem}.

\begin{remark}
\cref{assumption:davis} contains technical conditions about the choices made for the stochastic subgradient method.
We comment on how these can be met in \cref{remark:satisfying-davis}.

Condition (3) is often easy to meet; for example, by the second main result of \cite{wilkie}, this holds automatically if the functions consist of combinations of algebraic and exponential functions (see Section~3.3 of \cite{Carriere2021a} for examples for $\Phi$ that are particularly relevant in topological optimization).

The theoretical results of this paper are specifically about meeting the conditions involving the topological descriptor $\alpha$: assumptions (1) and (2).
\end{remark}

Finally, we comment on how to satisfy the standard assumptions for the stochastic subgradient method (\cref{assumption:davis}).

\begin{remark}
    \label{remark:satisfying-davis}
    Condition~(a) is immediate to satisfy, since we have control over the choice of learning rate.
    Condition~(c) is satisfied, for example, by any sequence of random variables $\{\zeta_k\}$ with zero mean, bounded variance, independent, and independent of $\{x_k\}$ and $\{y_k\}$. 
    Although condition~(b) seems more difficult to satisfy, as it involves proving something about the behavior of the stochastic subgradient method, it is known that it can be dealt with by adding a suitable regularization (for example, in the form of a non-topological loss $\Lscr'$ as in \cref{remark:generalizations-of-main-theorem}) forcing the iterates of the stochastic subgradient method to remain in a bounded region, or simply by restricting the domain of $\Phi$; we refer the interested reader to Section~6.1 of \cite{davis-drusvyatskiy-kakade-lee}.
\end{remark}

\section{Proofs of \cref{section:theory}}

\subsection{Proofs of \cref{section:stratifying}}

\begin{proof}[Proof of \cref{lemma:cells-one-parameter}]
    The function $(\Rbb^n)^K \to (\Rbb^K)^n$ in the statement is clearly a linear diffeomorphism, since its inverse is given by mapping $(g_1, \dots, g_n)$ to the function $K \to \Rbb^n$ given by sending $\sigma$ to $(g_1(\sigma), \dots, g_n(\sigma))$.
    It follows directly from the definition of cells that this function restricts to a bijection $\cell(f) \cong \cell(f_1) \times \cdots \times \cell(f_n)$, concluding the proof.
\end{proof}

\begin{lemma}
    \label{lemma:finitely-many-cells}
    The set $\{\cell(f) : f \in \Fil_n(K)\}$ is finite.
\end{lemma}

\begin{proof}
    Given $f \in \Fil_n(K)$, we have $\grid_f = [m_1] \times \cdots \times [m_n]$ with $m_i = |f_i(K)|$, so $m_i \leq |K|$.
    Thus, there are finitely many product posets $\grid_f$ that can be obtained.
    Moreover, there are finitely many possible maps $K \to \grid_f$ into each of these finitely many posets.
\end{proof}


\begin{proof}[Proof of \cref{lemma:cells-semialgebraic}]
    The first statement follows from the second and \cref{lemma:finitely-many-cells}, since distinct cells are disjoint, the union of all cells is exactly $\Fil_n(K)$, and finite unions
    of semilinear sets are semilinear.
    So it remains to show that all cells are semilinear.
    By \cref{lemma:cells-one-parameter}, it is sufficient to prove it in the one-parameter case, that is for $f \in \Fil_1(K)$.

    Given $\sigma,\tau \in K$, define a polynomial $P^{\sigma,\tau} : \Rbb^K \to \Rbb$ as $P^{\sigma,\tau}(h) = h_{\sigma}-h_{\tau}$, and a set
    $S^{\sigma,\tau} \subseteq \Rbb^K$ as follows:
    \[
     S^{\sigma,\tau} = 
        \begin{cases}
        \{h \in (\Rbb^n)^K : P^{\sigma,\tau}(h) = 0\}  & \text{ if $f(\sigma) = f(\tau)$}\\
        \{h \in (\Rbb^n)^K : P^{\sigma,\tau}(h) < 0\} & \text{ if $f(\sigma) < f(\tau)$}\\
        \{h \in (\Rbb^n)^K : P^{\sigma,\tau}(h) > 0\} & \text{ if $f(\sigma) > f(\tau)$.}\\
        \end{cases}
    \]
    Note that $\ord_{h} = \ord_{f}$ if and only if the preorder on the simplices of $K$ induced by $h$ is equal to that induced by $f$, and, in turn, this is true if and only if $h \in \bigcap_{\sigma,\tau \in K} S^{\sigma,\tau}$.
    Thus, $\cell(f) = \bigcap_{\sigma,\tau \in K} S^{\sigma,\tau}$, which is a finite intersection of semilinear sets, and thus semilinear.
\end{proof}

\begin{proof}[Proof of \cref{lemma:incl-open-semilinear}]
    We prove that $\incl([m]) \subseteq \Rbb^m$ is open, convex, and semilinear for every $m \geq 1 \in \Nbb$; this is sufficient, since $\incl(\Gcal) = \incl(\Gcal_1) \times \cdots \times \incl(\Gcal_n)$ can then be written as a finite intersection of open, convex, and semilinear sets.

    The fact that $\incl([m]) \subseteq \Rbb^m$ is open follows directly from its definition, which uses strict inequalities.
    Convexity follows from the fact that $r_1 < \cdots < r_n$ and $s_1 < \cdots < s_n$ implies $a r_1 + b s_1 < \cdots < a r_n + b s_n$ for any $a,b\geq 0 \in \Rbb$.
    We reduce the semilinearity statement to \cref{lemma:cells-semialgebraic}, as follows.
    Let $K$ be the simplicial complex on the set $\{0, \dots, m-1\}$ that has one $0$-simplex for each element and no higher-dimensional simplices.
    Consider the filtering function $f \in \Fil_1(K)$ with $f(i) = i$, for $i \in \{0, \dots, m-1\}$.
    Then, a straightforward check shows that $\incl([m]) = \cell(f)$, so the result follows from \cref{lemma:cells-semialgebraic}.
\end{proof}

\begin{proof}[Proof of \cref{lemma:carrier-iso}]
    Since everything is defined componentwise, it is sufficient to prove this for $f \in \Fil_1(K)$.

    The fact that the function is linear is evident from its definition; it thus suffices to prove that it is a bijection.
    Consider the function $F : \incl(\grid_f) \to \cell(f)$ mapping any injective and monotonic function  $\kappa : \grid_f \to \Rbb$ to $\kappa \circ \ord_f$.
    We prove that $F$ is the inverse of $C^*$.

    \emph{Right inverse.} If $\kappa \in \incl(\grid_f)$, then $C^*(F(\kappa))= \kappa \circ \ord_f \circ C = \kappa$, by the fact that $C$ is a carrier for $f$.

    \emph{Left inverse.} If $g \in \cell(f)$, then $F(C^*(g)) = g \circ C \circ \ord_f = \iota_g \circ \ord_g \circ C \circ \ord_f = \iota_g \circ \ord_f \circ C \circ \ord_f = \iota_g \circ \ord_f = \iota_g \circ \ord_g = g$, using the fact that $g = \iota_g \circ \ord_g$, the fact that $\ord_g = \ord_f$ twice, and the fact that $C$ is a carrier for~$f$.
\end{proof}

\subsection{Proofs of \cref{section:grid-determined}}
\begin{proof}[Proof of \cref{propositon:push-existence}]
    We start by proving existence.
    Consider the function $\push^\alpha_h : \incl(\Gcal) \to A$ given by mapping $\kappa \in \incl(\Gcal)$ to $\alpha(\kappa \circ h) \in A$.
    If $f \in \Fil_n(K)$ with $\grid_f = \Gcal$ and $\ord_h = h$ and $g \in \cell(f)$, we have $\push^\alpha_h(\iota_g) = \alpha(\iota_g \circ h) = \alpha(\iota_g \circ \ord_f) = \alpha(\iota_g \circ \ord_g) = \alpha(g)$.
    This shows existence.

    We now show uniqueness.
    Suppose $p : \incl(\Gcal) \to A$ also satisfies that, if $f \in \Fil_n(K)$ with $\grid_f = \Gcal$ and $\ord_f = h$, then $\alpha = p \circ \itsgridincl$.
    Note that, if $\kappa \in \incl(\Gcal)$ and $d = \kappa \circ h$, then since $h_i : K \to \Gcal_i$ is surjective for all $1 \leq i \leq n$, it follows that $\grid_d = \Gcal$ and $\ord_d = h$.
    Let $f \in \Fil_n(K)$ such that $\grid_f = \Gcal$ and $\ord_f = h$ (note that such a function always exists by taking, e.g., $f = \kappa \circ h$ for any $\kappa \in \incl(\Gcal)$).
    Now, if $\kappa \in \incl(\Gcal)$, then $p(\kappa) = p(\iota_{\kappa \circ h}) = \alpha(\kappa \circ h) = \push^\alpha_h(\kappa)$, as required.
\end{proof}

\begin{proof}[Proof of \cref{corollary:cells-nice}]
    It's enough to prove that $\itsgridincl = C^* \circ I_f : \cell(f) \to \incl(\grid_f)$; the second claim follows then follows directly from \cref{lemma:cells-one-parameter,lemma:incl-open-semilinear,lemma:carrier-iso}.

    Given $g \in \cell(f)$ and $x \in \grid_f$, we compute as follows
    \begin{align*}
        \iota_g(x) &= \iota_g(x_1, \dots, x_n) \\
            &= \iota_g\big(\ord_{f_1}(C_1(x_1)), \dots, \ord_{f_n}(C_n(x_n))\big)\\
            &= \iota_g\big(\ord_{g_1}(C_1(x_1)), \dots, \ord_{g_n}(C_n(x_n))\big)\\
            &= \big(\iota_{g_1}(\ord_{g_1}(C_1(x_1))), \dots, \iota_{g_n}(\ord_{g_n}(C_n(x_n)))\big)\\
            &= \big(g_1(C_1(x_1)) \dots, g_n(C_n(x_n))\big)\\
            &= I_f(g)(C(x))\\
            &= C^*(I_f(g))(x).
    \end{align*}
    Thus, $\iota_g = C^*(I_f(g)) : \grid_f \to \Rbb^n$, as required.
\end{proof}

\section{Semilinearity of known invariants}
\label{section:semilinearity-of-known-invariants}

We start by providing a proof of \cref{theorem:semilinearity-main-theorem}, which is a direct consequence of the results in this section.

The sorted Hilbert decomposition is defined in \cref{definition:sorted-hilbert}, and the evaluated multiparameter persistence landscape is defined in \cref{definition:evaluated-landscape}.

\begin{proof}[Proof of \cref{theorem:semilinearity-main-theorem}]
    This is a direct consequence of \cref{corollary:Hilbert-and-rank-semilinear} and \cref{proposition:landscape-semilinear}.
\end{proof}

\subsection{Semilinearity of Hilbert decomposition signed measure}
\label{section:hilbert-rank-semilinear}


We do the case of the Hilbert decomposition signed measure, and thus that of discrete signed measures on $\Rbb^n$; the case of the rank decomposition signed measure is only slightly more verbose, but not conceptually harder.

In order to prove semilinearity, we must represent the Hilbert decomposition signed measures as elements of a finite dimensional vector space.

For notational clarity, fix $n \geq 1 \in \Nbb$, $i \in \Nbb$, and a simplicial complex $K$.

In \cref{section:constant-number-of-masses}, we prove a technical result stating that the Hilbert decomposition signed measure has a constant number of positive and negative masses on each cell.
In \cref{section:hilbert-semilinearity}, we prove the semilinearity of the sorted Hilbert decomposition.

\subsubsection{Constant number of point masses}
\label{section:constant-number-of-masses}


It is convenient to also consider discrete signed measures on discrete sets, such as grids.
In this case, we do not define optimal transport distances, since we do not make use of these.

\begin{lemma}
    \label{lemma:push-forward-on-cell}
    Let $i \in \Nbb$.
    Let $\Gcal = \Gcal_1 \times \cdots \times \Gcal_n$ be a finite grid and let $h \in \Fil_\Gcal(K)$.
    There exists a signed measure $\mu^\Hil_{H_i(h)} \in \dsm(\Gcal)$ such that, for every grid inclusion $\kappa \in \incl(\Gcal)$ we have
    \[
        \mu^\Hil_{H_i(\kappa \circ h)} = 
        \kappa_\#\, \mu^\Hil_{H_i(h)} \; \in \; \dsm(\Rbb^n),
    \]
    where $\kappa_\# : \dsm(\Gcal) \to \dsm(\Rbb^n)$ denotes the push forward of measures induced by the inclusion $\kappa : \Gcal \to \Rbb^n$.
\end{lemma}
\begin{proof}
    The existence of $\mu^\Hil_{H_i(h)} \in \dsm(\Gcal)$ follows from the existence of Hilbert decompositions \cite{oudot-scoccola}.
    The equality between measures follows from the uniqueness of $\mu^\Hil_{H_i(\kappa \circ h)}$.
\end{proof}

Note that push forwards of discrete measures by any function (such as $\kappa$ in \cref{lemma:push-forward-on-cell}) have a very simple expression: all one needs to do is to apply the function to the coordinates of each of the point masses.


\begin{lemma}
    \label{lemma:constant-number-of-masses}
    Let $f \in \Fil_n(K)$.
    There exists $p_f, q_f \in \Nbb$ such that, for every $g \in \cell(f)$, the number of positive (resp.~negative) point masses of $\mu^\Hil_{H_i(g)}$ is $p_f$ (resp.~$q_f$).
\end{lemma}
\begin{proof}
    This follows immediately from 
    \cref{lemma:push-forward-on-cell} and the fact that $\kappa$ is injective, so that no positive-negative pair of masses cancels.
\end{proof}

\subsubsection{Proof of semilinearity}
\label{section:hilbert-semilinearity}

\cref{lemma:constant-number-of-masses} allows us to index the coordinates of the point masses of $\mu^\Hil_{H_i(g)}$ in a convenient way, as follows.

\begin{construction}
    \label{construction:ordered-hilbert}
Let $C \in \Nbb$ be the number of cells of $\Fil_n(K)$, which is finite by \cref{lemma:finitely-many-cells}.
Choose representatives $f^1, \dots, f^C \in \Fil_n(K)$ of the cells of $\Fil_n(K)$.
Consider the finite set
\[
    \indexingset' \coloneqq
    \left\{
        \begin{array}{c}
        (c, s, j) \in \Nbb \times \{+,-\} \times \Nbb, \text{ such that}\\
        1 \leq c \leq C,\\
        1 \leq j \leq p_{f^c},\; \text{ if $s = +$}\\
        1 \leq j \leq q_{f^c},\; \text{ if $s = -$}\\
        \end{array}
    \right\},
\]
and define the disjoint union $\indexingset \coloneqq \{1, \dots, C\} \amalg \indexingset'$.

Given $g \in \Fil_n(K)$, let $1 \leq e \leq C$ such that $g \in \cell(f^{e})$.
Using \cref{lemma:constant-number-of-masses}, consider the unique ordering $x_1, \dots, x_{p_{f^e}} \in \Rbb^n$ of the positive point masses of $\mu^{\Hil}_{H_i(g)}$ that is compatible with the lexicographic order of $\Rbb^n$.
Analogously, let $y_1, \dots, y_{q_{f^e}} \in \Rbb^n$ be the unique order of the negative point masses of $\mu^{\Hil}_{H_i(g)}$ that is compatible with the lexicographic order of $\Rbb^n$.

Consider the element $\ordHil(g) \in \Rbb^\indexingset$ defined by 
\[
    \ordHil(c) = 
        \begin{cases*}
            0 & \text{ if $c \neq e$}\\
            1 & \text{ if $c = e$},
        \end{cases*}
\]
when $c \in \{1, \dots, C\}$, and by 
\[
\ordHil(c,s,j) = 
    \begin{cases*}
    0 & \text{ if $c \neq e$}  \\
    x_j & \text{ if $c = e$ and $s = +$}\\
    y_j & \text{ if $c = e$ and $s = -$},
    \end{cases*}
\]
when $(c,s,j) \in \indexingset'$.
\end{construction}

In words, $\ordHil(g)$ contains the information of which cell $g$ belongs to, as well as the coordinates of the point masses of the Hilbert decomposition signed measure of $g$.

\begin{definition}
    \label{definition:sorted-hilbert}
    Let $i \in \Nbb$.
    Define the $i$th \emph{sorted Hilbert decomposition} as $\ordHil : \Fil_n(K) \to \Rbb^\indexingset$, as in \cref{construction:ordered-hilbert}.
\end{definition}

\begin{proposition}
    \label{corollary:Hilbert-and-rank-semilinear}
    The sorted Hilbert decomposition is semilinearly determined on grids, and thus semilinear.
\end{proposition}
\begin{proof}
    By \cref{theorem:grid-determined-implies-semialgebraic}, it is sufficient to let $f \in \Fil_n(K)$ be arbitrary, and show that the function $\push^\ordHil_{\ord_f} : \incl(\grid_f) \to \Rbb^\indexingset$ is semilinear.
    This function is in fact affine.

    To see this, we let $h = \ord_f$ and use \cref{lemma:push-forward-on-cell}.
    This guarantees that each point mass of $\mu^\Hil_{H_i(g)}$ lies on top of $\iota_g(x)$ for some element $x \in \grid_f$, and thus varies affinely with respect to $\iota_g \in \incl(\grid_f)$.
    Moreover, the lexicographic order of the point masses of $\mu^\Hil_{H_i(g)}$ does not change as $g$ varies: the order on $\grid_f$ induced by pulling back the lexicographic order of $\Rbb^n$ along any inclusion $\iota_g$ is the same for all $g \in \cell(f)$.
    This implies that $\ordHil(g)$ varies affinely with respect to $\incl_g$ as long as $g \in \cell(f)$, as required.
\end{proof}

\subsection{Semilinearity of the multiparameter persistence landscape}

For ease of notation, let us fix a simplicial complex $K$, a number of parameter $n \geq 1 \in \Nbb$, and a homological dimension $i \in \Nbb$.

We now recall the definition of multiparameter persistence landscape from \cite{Vipond2020}.
Note that \cite{Vipond2020} works with general multiparameter persistence modules; we specialize the definition to the case of a homology multiparameter persistence module.

\begin{definition}
Let $f \in \Fil_n(K)$.
Let $k \geq 1 \in \Nbb$.
The $k$th \emph{multiparameter persistence landscape} of $H_i(f)$ is the function
$\lambda^k_f : \Rbb^n \to \Rbb$ defined by
\[
    \lambda^k_f(z) \coloneqq \sup\left\{ \epsilon \geq 0 \in \Rbb : \rk_f(z-h,z+h) \geq k \text{ for all $h \in \Rbb^n_{\geq 0}$ with $\|h\|_\infty \leq \epsilon$ }\right\} \in \Rbb,
\]
where we write $\rk_f = \rk(H_i(f))$, to simplify notation.
\end{definition}

By convention, the supremum of an empty set is taken to be zero.
Note that this supremum is always finite since the support of $H_i(f) : \Rbb^n \to \vect$ is bounded from below.

As a function of $f \in \Fil_n(K)$, the multiparameter landscape takes values in a space of functions $\Rbb^n \to \Rbb$, which is not a finite dimensional vector space.
In order to prove a semilinearity result, we consider the evaluation of the landscape at points of $\Rbb^n$, and show that this is semilinear.

\begin{definition}
    \label{definition:evaluated-landscape}
    Let $z_0 \in \Rbb^n$.
    The \emph{evaluated multiparameter persistence landscape} is the descriptor $\Fil_n(K) \to \Rbb$ mapping $f$ to $\lambda^k_f(z_0)$.
\end{definition}

In fact, we prove the following stronger result, which, down the line, allows for the optimization of the points over which the landscape is evaluated.

\begin{proposition}
    \label{proposition:landscape-semilinear}
    The evaluated multiparameter persistence landscape is linearly determined on grids, and thus semilinear.
    Moreover, the function $\Fil_n(K) \times \Rbb^n \to \Rbb$ mapping $(f,z)$ to $\lambda^k_f(z)$ is semilinear.
\end{proposition}

To simplify notation even more, let us fix $k \geq 1 \in \Nbb$ and let $\lambda_f \coloneqq \lambda^k_f$.

Let $\mathbf{1} = (1, \dots, 1) \in \Rbb^n$.
Multiparameter landscapes can be computed using lines of slope $\mathbf{1}$ as follows (see Lemma~21 of \cite{Vipond2020}): $\lambda_f(z) = \sup\left\{ \epsilon \geq 0 : \rk_f(z-\mathbf{1}\epsilon,z+\mathbf{1}\epsilon) \geq k \right\}$.
In particular, we get the following, by the monotonicity of the rank:
\begin{equation}
    \label{lemma:reduction-to-slope-1}
    \lambda_f(z) = \sup\left\{ \min(r,s) : r,s \geq 0 \in \Rbb, \rk_f(z-\mathbf{1}r,z+\mathbf{1}s) \geq k \right\}.
\end{equation}

We now give some useful constructions for the proof of \cref{proposition:landscape-semilinear}.
Given $x,z \in \Rbb^n$ and $y \in \Rbb^n \cup \{\infty\}$, let
\begin{align*}
    d_\downarrow(z,x) &\coloneqq
        \begin{cases*}
            0 & \text{ if $x \nleq z$}\\
            \min\limits_{1 \leq j \leq n} |z_j - x_j| & \text{ if $x \leq z$}
        \end{cases*}\\
    d_\uparrow(z,y) &\coloneqq
        \begin{cases*}
            \infty & \text{ if $y = \infty$}\\
            0 & \text{ if $z \geq y \in \Rbb^n$}\\
            \max\limits_{\substack{
                1 \leq j \leq n\\
                \text{s.t. } z_j \leq y_j
                }} |z_j - y_j| & \text{ if $z \ngeq y \in \Rbb^n$}
        \end{cases*}\\
    d_\updownarrow(z,x,y) &\coloneqq \min\left(d_\downarrow(z,x), d_\uparrow(z,y)\right).
\end{align*}

Given $x,z \in \Rbb^n$, $y \in \Rbb^n \cup \{\infty\}$, with $x \leq z$ and $z \ngeq y$, let
\begin{align*}
    z\downarrow x &\coloneqq z - \mathbf{1}d_\downarrow(x,z) \in \Rbb^n\\
    z\uparrow y &\coloneqq z + \mathbf{1}d_\uparrow(z,x)
        \in \Rbb^\infty \cup \{\infty\},
\end{align*}
which have the following property:
if $U_w = \{u \in \Rbb^n : u \geq w\} \subseteq \Rbb^n$, then $z\downarrow x$ is the intersection between the line $r \mapsto z + \mathbf{1}r$ and $\partial U_x$, the boundary of $U_x$; and $z\uparrow y$ is the intersection between the line $r \mapsto z + \mathbf{1}r$ and $\partial U_y$.
Note that, by definition, we have $\|z - (z\downarrow x)\|_\infty = d_\downarrow(z,x)$ and similarly $\|z - (z\uparrow y)\|_\infty = d_\downarrow(z,y)$.

See \cref{figure:landscape-projection} for an illustration.

\begin{figure}
    \centering
    \includegraphics[width=6cm]{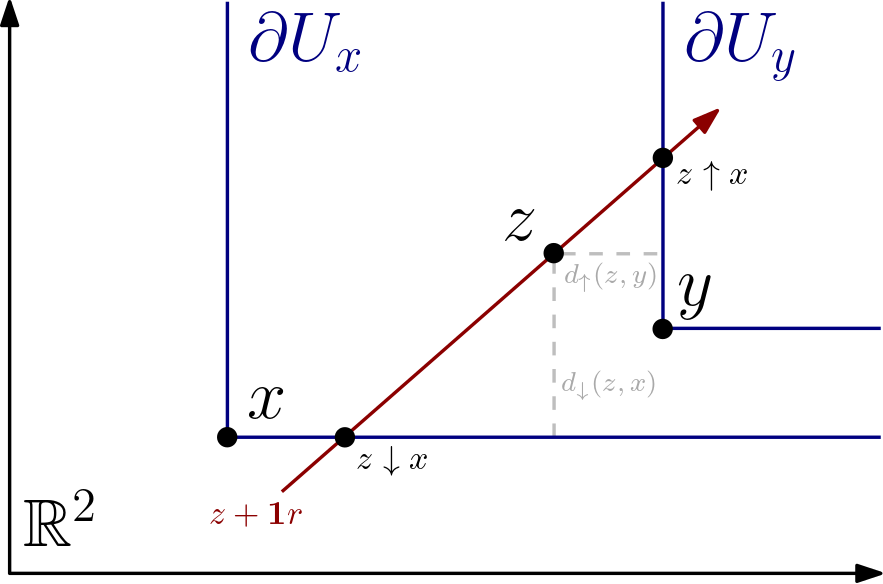}
    \caption{Illustration of the main constructions involved in \cref{proposition:landscape-semilinear}.}
    \label{figure:landscape-projection}
\end{figure}

It is standard (and easy to see) that $r \mapsto \rk_f(z-\mathbf{1}r,z+\mathbf{1}r)$ is constant as long as $z-\mathbf{1}r$ and $z+\mathbf{1}r$ do not cross any of the coordinates of the inclusion of the grid of $f$ in $\Rbb^n$; more precisely, as long as $z-\mathbf{1}r$ and $z+\mathbf{1}r$ do not cross the boundary of $U_{\iota_f(a)}$ for some $a \in \grid_f$.
From this observation and \cref{lemma:reduction-to-slope-1}, we get
\begin{equation}
    \label{equation:landscape-finitely-many}
    \lambda_f(z) = \max \left\{ d_\updownarrow(z,\iota_f(a),\iota_f(b)) : a < b \in \overline{\grid_f}, \rk_{\ord_f}(a,b) \geq k\right\},
\end{equation}
where $\overline{\grid_f} = \grid_f \cup \{\infty\}$, the function $\rk_{\ord_f} : \{(a,b) : a < b \in \overline{\grid_f}\} \to \Zbb$ denotes the rank function of the $\grid_f$ filtration $\ord_f : K \to \grid_f$, extended as $\rk_{\ord_f}(a, \infty) \coloneqq \rk_{\ord_f}(a,\max \grid_f)$, and also extending $\iota_f(\infty) = \infty \in \Rbb^n \cup \{\infty\}$.
The extension of $\grid_f$ to $\overline{\grid_f}$ is a minor technical point required to handle cases in which a rank does not go below $k$ as its second coordinate goes to $\infty$, i.e., the case in which $\lambda_f(z) = d_\downarrow(z,\iota_f(a))$ and $\rk_f(z \downarrow \iota_f(a), z + \mathbf{1}r) \geq k$ for all $r \geq 0 \in \Rbb$.

Since $\grid_f$ is finite, the supremum of \cref{lemma:reduction-to-slope-1} is now a maximum; and again we use the convention that a maximum over an empty set is zero.

\begin{proof}[Proof of \cref{proposition:landscape-semilinear}]
    Let $f \in \Fil_n(K)$, and let $g \in \cell(f)$.
    Since $\grid_g = \grid_f$ and $\ord_g = \ord_f$, we have, by \cref{equation:landscape-finitely-many},
    \[
        \lambda_g(z) = \max \left\{ d_\updownarrow(z,\iota_g(a),\iota_g(b)) : a \leq b \in \grid_f, \rk_{\ord_f}(a,b) \geq k\right\},
    \]
    which is a maximum of finitely many semilinear functions that do not depend on $g$ or $z$, and hence semilinear in both $z$ and $\iota_g$.
    It follows that the function $\Fil_n(K) \times \Rbb^n \to \Rbb$ mapping $(g,z)$ to $\lambda_g(z)$ is semilinear when restricted to $\cell(f) \times \Rbb^n$ for each $f \in \Fil_n(K)$, and thus semilinear.
    For the first claim, note that $\push_{\ord_f}(\iota_g) = \lambda_g(z)$, and use \cref{theorem:grid-determined-implies-semialgebraic}.
\end{proof}

\section{Locally Lipschitz objective functions}
\label{section:locally-lipschitz-objective}



In this section we prove \cref{theorem:satisfy-locally-lipschitz-condition}, which gives simple conditions under which an objective relying on the Hilbert decomposition signed measure is locally Lipschitz. 

In \cref{section:algebraic-bound}, we prove a bound that allows us to deduce the stability of signed barcodes as measures descriptors.
In \cref{section:locally-lipschitz-signed-measures}, we use this to give sufficient conditions for a Hilbert decomposition signed measure-based loss function to be locally Lipschitz.

\subsection{An algebraic bound}
\label{section:algebraic-bound}

Since this is the only section containing algebraic arguments, and these are encapsulated and only used in the stability result \cref{proposition:hilbert-decomposition-lipschitz}, we refer the reader to, e.g., \cite{oudot-scoccola,botnan-lesnick} for background.

\begin{lemma}
    \label{lemma:main-algebraic-bound}
    Fix a finite simplicial complex $K$ and $n \in \Nbb$.
    There exists a constant $B$ that only depends on $K$ and $n$ such that, for every $f \in \Fil_n(K)$ and $i \in \Nbb$, the sizes of the Betti signed barcode \cite{oudot-scoccola} and of the rank exact decomposition \cite{botnan-oppermann-oudot-scoccola} of the multiparameter persistence module $H_i(f) : \Rbb^n \to \vect$ are bounded above by~$B$.
\end{lemma}
\begin{proof}
    We first prove the following.

    \emph{Claim a.}
    Fix $n \in \Nbb$.
    Let $P \to Q$ be a map of free $n$-parameter persistence modules, of ranks $p,q \in \Nbb$, respectively. 
    There exists a constant $B'$ only depending on $n,p,q\in\Nbb$ (and not on $P$, $Q$, or the map between them), such that $\bsf(\ker(P \to Q)) \leq B'$.

    \emph{Proof of Claim a.}
    This is a consequence of the main result of \cite{beecher}, in the case where the chosen presentation is not minimal.

    \smallskip

    Now, for any fixed $f \in \Fil_n(K)$, let $C^f_\bullet$ denote the associated chain complex of $n$-parameter persistence modules, whose homology is the family of homology persistence modules $H_\bullet(f)$.
    Note that, for every $k \in \Nbb$, the rank of the free module $C^f_k$ is independent of $f$.
    By definition of homology, we have $H_i(f) = \coker(C^f_{i+1} \to Z^f_i)$, where $Z^f_i$ is the kernel of the boundary morphism $Z^f_i = \ker(d^f_i : C^f_i \to C^f_{i-1})$.

    To simplify notation, let $\bsf(M)$ denote the size of the Betti signed barcode of a multiparameter persistence module $M$, as in \cite{botnan-oppermann-oudot-scoccola}.

    \smallskip

    \emph{Claim b.}
    There exists a constant $B''$, only depending on $K$ and $n$ (and not on $f$), such that, for all $k \in \Nbb$, $\bsf(Z^f_k) \leq B''$.

    \emph{Proof of Claim b.}
    Since $C^f_k$ is non-zero only for finitely many $k \in \Nbb$ (because $K$ is finite), this follows from Claim a and the fact that $Z^f_i = \ker(C^f_i \to C^f_{i-1})$.

    \smallskip

    We now prove that $\bsf(H_i(f))$ is bounded above by a constant that only depends on $K$ and $n$.
    The bound for the rank exact decomposition of $H_i(f)$ follows from this and Proposition~5.28 of \cite{botnan-oppermann-oudot-scoccola}.

    Note that $\ker(C_{i+1} \to Z_i) \cong \ker(C_{i+1} \to C_i) = Z_{i+1}$, so we have a short exact sequence 
    \[
      0 \to C_{i+1}/Z_{i+1} \to Z_i \to H_i(f) \to 0
    \]
    By Lemma~5.21 of \cite{botnan-oppermann-oudot-scoccola}, to bound $\bsf(H_i(f))$, it is thus enough to bound $\bsf(C_{i+1}/Z_{i+1})$ and $\bsf(Z_i)$ by a constant that only depends on $K$ and $n$.
    To bound $\bsf(Z_i)$ we use Claim b.

    To bound $\bsf(C_{i+1}/Z_{i+1})$, we use 
    Lemma~5.21 of \cite{botnan-oppermann-oudot-scoccola}, this time with the short exact sequence $0 \to Z_{i+1} \to C_{i+1} \to C_{i+1}/Z_{i+1} \to 0$, and Claim b.
    This concludes the proof.
\end{proof}

As a direct consequence of
\cref{lemma:main-algebraic-bound}, we get the following.

\begin{corollary}
    Fix a finite simplicial complex $K$ and $n \in \Nbb$.
    There exists a constant $B$ that only depends on $K$ and $n$ such that, for every $f \in \Fil_n(K)$ and $i \in \Nbb$, the number of point masses in both $\mu^\Hil_{H_i(f)}$ and $\mu^\rk_{H_i(f)}$ is bounded above by $B$.
    \qed
\end{corollary}

\subsection{Locally Lipschitz signed measure-based objective functions}
\label{section:locally-lipschitz-signed-measures}

The next result follows easily from the main results of \cite{oudot-scoccola,botnan-oppermann-oudot-scoccola}, and the algebraic bound (\cref{lemma:main-algebraic-bound}).

\begin{proposition}
    \label{proposition:hilbert-decomposition-lipschitz}
    Let $K$ be a finite simplicial complex and let $n,i \in \Nbb$.
    The following functions are Lipschitz with respect to any $\ell^p$ norm on $\Fil_n(K) \subseteq (\Rbb^n)^K$ and the (partial) optimal transport distance on discrete signed measures:
    \begin{align*}
        \mu_{H_i}^\Hil &: \Fil_n(K) \to \dsm(\Rbb^n)\\
        \mu_{H_i}^\rk &: \Fil_n(K) \to \dsm\left(\bars_n, \Delta\right).
    \end{align*}
\end{proposition}
\begin{proof}
    From the bottleneck stability of the Betti signed barcode and the rank exact decomposition
    \cite{oudot-scoccola,botnan-oppermann-oudot-scoccola}
    we obtain a partial optimal transport  (\cref{definition:pot}) stability result
    (sometimes known as $1$-Wasserstein stability result in the topological data analysis literature), as long as have a global bound on the number of point masses.
    This is because the Hilbert decomposition signed measure is obtained from the Betti signed barcode by cancelling equal masses that appear as positive and as negative (Remark~5.3 of \cite{oudot-scoccola}); and the same is true for the rank exact decomposition and the rank decomposition signed barcode.
    The bound on the number of masses is the content of \cref{lemma:main-algebraic-bound}.
\end{proof}

Again, for notational clarity, fix $n \geq 1 \in \Nbb$, $i \in \Nbb$, and a simplicial complex $K$.

Recall that $\ordHil$ denotes the sorted Hilbert decomposition (\cref{definition:sorted-hilbert}), which is a convenient representation of the Hilbert decomposition signed measure as a vector of a finite dimensional space.

\begin{proposition}
    \label{theorem:satisfy-locally-lipschitz-condition}
    Assume given a locally Lipschitz function $E' : \dsm(\Rbb^n) \to \Rbb$, and let $E : \Rbb^\indexingset \to \Rbb$ be any function such that $E' \circ \mu^\Hil_{H_i} = E \circ \ordHil : \Fil_n(K) \to \Rbb$.
    Then, $E\circ \ordHil : \Fil_n(K) \to \Rbb$ is locally Lipschitz.
\end{proposition}
\begin{proof}
    By assumption, $E \circ \ordHil = E' \circ \mu^{\Hil}_{H_i}$, and the right-hand side is a composite of a locally Lipschitz map $E'$, by assumption, and a Lipschitz map $\mu^{\Hil}_{H_i} : \Fil_n(K) \to \dsm(\Rbb^n)$, by \cref{proposition:hilbert-decomposition-lipschitz}.
\end{proof}

Note that a function $E$ as in \cref{theorem:satisfy-locally-lipschitz-condition} always exists since all $\ordHil(f)$ is doing is encoding the point masses of the signed measure $\mu^\Hil_{H_i}(f)$ in a convenient way.
In other words, $E$ is simply an explicit representation of $E'$ when the signed measure is encoded as in \cref{construction:ordered-hilbert}.

\section{Details about experiments}
\label{section:experiments_appendix}

\subsection{Details about autoencoder experiment}
\label{section:details-autoencoder}
We use a simple autoencoder architectures with both encoders and decoders made of three layers of 32 neurons, with the first two layers followed by ReLU activation and batch normalization.
We then minimize a linear combination of the MSE loss (between initial and reconstructed spaces) and the topological loss (\ref{eq:mp_autoencoder_loss}) (between initial and latent spaces), 
with weights $1$ and $0.1$ for the MSE and topological losses respectively, to account for the scale difference between the two losses. Optimization is performed with Adam optimizer, learning
rate $0.01$, and $1000$ epochs.

\subsection{Details about graph data example}
\label{subsec:details-graph}

In order to create node and edge attributes from persistence diagrams%
\footnote{I.e., one-parameter barcodes, or one-parameter rank decomposition signed measure $\mu^\rk$ (\cref{section:rank-dec-signed-measure}), in our language.}%
, the authors in~\cite{Horn2022} propose to use the natural bijections between persistence diagram points in dimension 0 (resp. dimension 1)
with nodes (resp. edges) of the graph\footnote{Note that, in dimension 1, this bijection is only well defined after matching the edges that are not involved in the creation of any graph cycle
to an arbitrary persistence diagram point on the diagonal---see Section A.4 in~\cite{Horn2022}}. These bijections can be used to permute the persistence diagrams, so that every node (resp. edge) can be associated to its $k$
corresponding persistence diagram points in dimension 0 (resp. dimension 1), and further processed with, e.g., a DeepSet architecture, in order to create a single node vector that is added to the one
obtained from graph convolutions. The edge vectors, on the other hand, are pooled so as to create graph-level descriptors. 

Unfortunately there is no such correspondence between nodes or edges and the Hilbert decomposition signed measure in dimension 0 and 1. Hence, for a given node $v$ with associated filtration values $f(v)\in\R^k$ (again,
recall that $f$ is obtained from graph convolutions), we create node vectors in $\R^m$ by computing:

\begin{equation}\label{eq:mpvect}
v \mapsto \int \psi_{v,i}\, \dsf\mu^\Hil_{H_*(f)}, 
1\leq i \leq m,
\end{equation}

where $\psi_{v,i}$ is a function of the form $\psi_{v,i}(p)={\rm exp}\left(-(p-f(v))^T \Sigma^{-1}_i (p-f(v))\right)$,
and the $\Sigma_i$'s are $m$ learnable SPD matrices in $\R^{k\times k}$.
This is thus an instance of~\cref{example:integration-function}.
These vectors are then mapped back to $\R^k$ with a fully connected layer in order to add them to the original node attributes, as proposed in~\cite{Horn2022}. We use the same procedure for edges, except
that the edge vectors in $\R^m$ are pooled into graph-level descriptors.

For graph neural networks, we use Graph Convolutional Networks (GCN)~\cite{Kipf:2016tc},
Graph Isomorphism Network (GIN)~\cite{Xu2019},
Graph ResNet~\cite{7780459},
and Graph DenseNet~\cite{huang2017densely},
All graph architectures have four layers with 256 neurons, and, whenever topological vectors
are used, they are placed after the second layer and computed from $k=2$ filtrations.
In our experience, using a larger $k$ produced comparable results at a higher computational cost.
GNNs are trained during 200 epochs with Adam optimizer with learning rate 0.001, and performance is computed over 10 train/test folds.


\begin{figure}
    \centering
    \includegraphics[width=8.5cm]{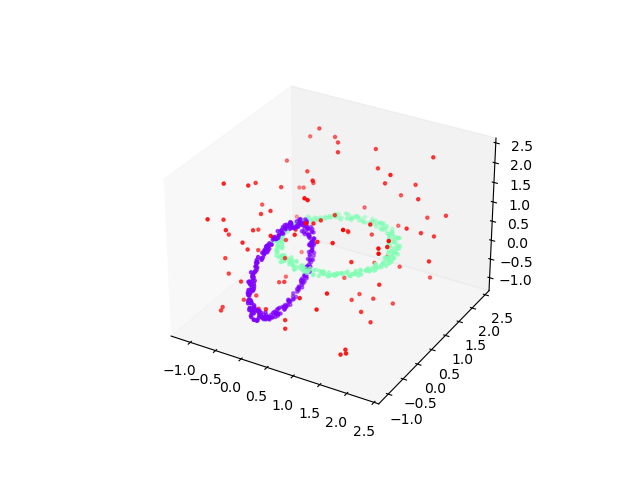}
    \includegraphics[width=8.5cm]{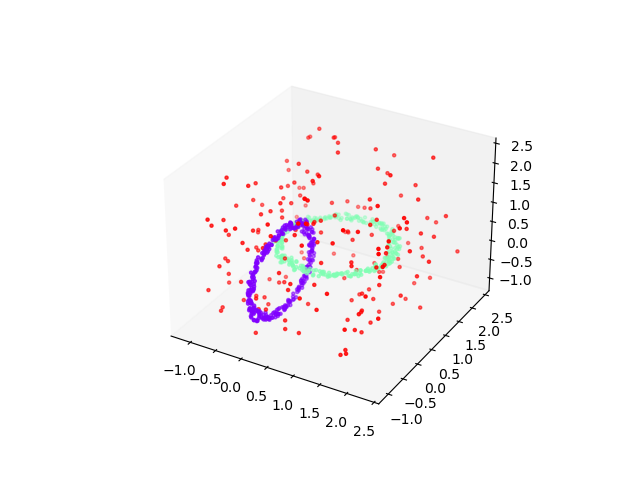}
\caption{
Two point cloud datasets consisting of two interlaced circles with background noise, embedded in $\R^9$, similar to the data used in~\cite{Carriere2021a}. 
The difference between the two datasets is the amount of background noise.
}
\label{figure:autoencoder-data}
\end{figure}

\newpage

\section{Dependency graph and notation table}

\null\vfill

\begin{figure}[H]
\centering
\includegraphics[width=17cm]{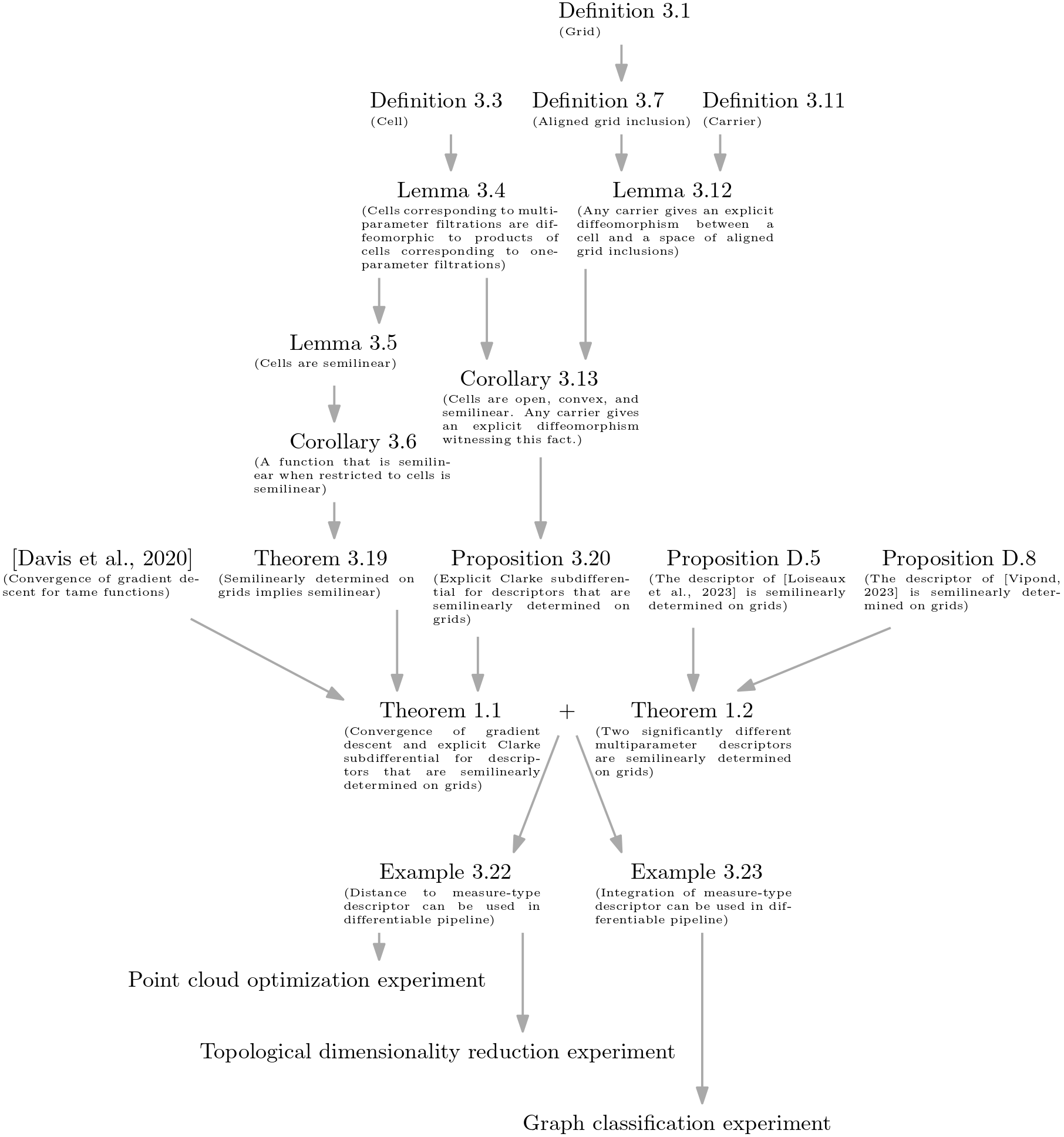}
\caption{Dependency graph of the main definitions, results, and experiments in the paper.}
\label{figure:dependency-graph}
\end{figure}
\vfill\null

\newpage

\begin{table}[H]
    \centering
    \vspace{10pt} 
    \begin{tabular}{|p{6cm}|p{10cm}|} 
    \hline
    \textbf{Symbol} & \textbf{Description} \\
    \hline
    $K$ & A \emph{finite simplicial complex} (\autopageref{simp-comp}). \\
    $(\Rbb^n)^X$ & The set of functions $X \to \Rbb^n$ (\autopageref{functions-finite-set}). \\
    $\Fil_n(K) \subseteq (\Rbb^n)^K$ & The \emph{set of $n$-filtrations} of $K$ (\cref{definition:n-filtrations}). \\
    $\mathbb{F}$ & A field. \\
    $\Phi: \mathbb{R}^d \to \Fil_n(K)$ & A \emph{parameterized family of filtrations}. \\
    $H_i(K)$ & The $i$th \emph{homology} of $K$ with coefficients in $\Fbb$.\\
    $H_i(f)$ & The $i$th \emph{multiparameter persistent homology} of an $n$-filtration $f$ of ${K}$ (\cref{definition:multipara-pers-hom}). \\
    $\Hil(H_i(f)) : \Rbb^n \to \Zbb$  & The $i$th \emph{Hilbert function} of $f \in \Fil_n(K)$ (\autopageref{hil-func}).\\
    ${\rk(H_i(f)) : \{(r,s) \in (\Rbb^n)^2 : r \leq s\} \to \Zbb}$ &The $i$th \emph{rank invariant} of $f \in \Fil_n(K)$ (\autopageref{rank-inv}).\\
    $M$ & A metric space. \\
    $\dm(M)$ & The set of \emph{discrete measures} on $M$ (\autopageref{disc-measures}).\\
    $\dsm(M)$ & The set of \emph{discrete signed measures} on $M$ (\autopageref{disc-measures}).\\
    $\OT$& The \emph{optimal transport distance} on the space of discrete measures (\cref{section:pot-definition}).\\
    $\mu^\Hil_{H_i(f)} \in \dsm(\Rbb^n)$  & The $i$th \emph{Hilbert decomposition signed measure} (\cref{definition:hilbert-sm-rk-sm}).\\
    $S\subset \mathbb{R}^n$ & A semilinear set (\autopageref{semilinear-sets}).\\
    $\partial \Lcal(z)$ & The \emph{Clarke subdifferential} of $\Lcal$ at $z \in \Rbb^d$ (\autopageref{clarke-subdiff}).\\
    $[m]$ & The \emph{linear order} $ \{0 < 1 < \cdots < m-1\}$ for ${m\in \Nbb}$. \\
    $\Gcal = \Gcal_1 \times \cdots \times \Gcal_n$ & A \emph{grid} (\cref{definition:grid}). \\
    $\grid_f=[m_1] \times \cdots \times [m_n]$ & The \emph{grid} induced by a filtration ${f}$ (\cref{construction:grid-from-filtration}). \\
    $\ord_f: K \to \grid_f$ & The \emph{unique linear preorder} induced by a filtration ${f}$ (\cref{construction:grid-from-filtration}). \\ 
    $\cell(f) \subseteq \Fil_n(K)$ & The \emph{cell} of $f \in \Fil_n(K)$ (\cref{definition:cell}).\\
    $\incl([m])$ & The set of \emph{aligned grid inclusions} of ${[m]}$ into ${\mathbb{R}}$ (\cref{construction:aligned-grid-inclusions-and-incl}).\\
    $\incl(\mathcal{G})$ & The set of aligned grid inclusions of a grid $\Gcal$ into $\Rbb^{[m_1]} \times \cdots \times \Rbb^{[m_n]}$ (\cref{construction:aligned-grid-inclusions-and-incl}).\\
    $C : \grid_f \to K$ & A \emph{carrier} for $f \in \Fil_1(K)$ (\cref{definition:carrier}).\\ 
    $\{C_i : \grid_{f_i} \to K\}_{1 \leq i \leq n}$ & A \emph{carrier} for a multifiltration $f \in \Fil_n(K)$ (\cref{definition:carrier}).\\
    $\itsgridincl : \cell(f) \to \incl(\grid_f)$ & Function mapping a filtration in the cell of $f$ to its corresponding grid inclusion (\cref{definition:iota}).\\
    $\Fil_\Gcal(K)$ & The \emph{set of $\Gcal$-filtrations} on $K$ (\autopageref{grid-filts}).\\
    $\VR(X)$ & The \emph{Vietoris--Rips filtration} of a finite point cloud ${X\subset \mathbb{R}^n}$.\\ 
    \hline
    \end{tabular}
    \caption{Notation Table.}
    \label{tab:notation-table}
\end{table}

\end{document}